\documentclass[aps,a4paper,12pt,nofootinbib]{revtex4} 

\usepackage{amsmath}

\usepackage{amsthm}
\usepackage{amssymb}
\usepackage[utf8]{inputenc}
\usepackage{graphicx}
\usepackage{hyperref}
\usepackage[arrow,curve,matrix]{xy}

\usepackage{color}
\usepackage{mathptmx}
\usepackage{amsmath}
\usepackage{amsfonts}
\usepackage{mathrsfs}
\usepackage[mathscr]{euscript}
\usepackage{mathtools}
\usepackage{textcomp} 
\usepackage{bm}
\usepackage{amssymb}

\newtheorem{theorem}{Theorem}
\newtheorem{lemma}[theorem]{Lemma}

\def\Real{{\mathbb R}}

\def\innerprod(#1,#2){{\left<#1\,,\,#2\right>}}
\def\Set#1{{\left\{#1\right\}}}
\def\qquadtext#1{\qquad\textup{#1}\qquad}
\def\qquadand{\qquadtext{and}}
\def\quadtext#1{\quad\textup{#1}\quad}
\def\quadand{\quadtext{and}}
\def\pfrac#1#2{\frac{\partial #1}{\partial #2}}
\def\ppfrac#1#2{\frac{\partial^2 #1}{\partial {#2}^2}}

\def\pqfrac#1#2#3{\frac{\partial^2 #1}{\partial {#2}\partial{#3}}}

\def\Mag{{M}}            
\def\Pol{{P}}            
\def\Efield{{E}}         
\def\Bfield{{B}}         
\def\Dfield{{D}}         
\def\Hfield{{H}}         
\def\Ffield{{F}}         
\def\Gfield{{G}}         
\def\Jvec{{J}}           

\def\FTepsilon{\epsilon}   

\def\Hdual{{\star}}      
\def\Star{{{\Hdual}\,}}
\def\StarX{{{\Star}_{{}_X}}}  
\def\Starhat{{\indevice{\Hdual}\,}}
\def\StarXhat{{\indevice{\Hdual}_{{}_X}}}

\def\pstar{{\ast}}  
\def\mstar{{-\ast}}

\newcommand{\vectorPart}[1]{\underline{#1}}

\def\Vx{{\vectorPart{x}}}
\def\Vy{{\vectorPart{y}}}
\def\PiX{{\pi_\textup{X}}}
\def\PiY{\pi_\textup{Y}}
\def\morphX{\morph_\textup{X}}
\def\morphY{\morph_\textup{Y}}
\def\morphXY{\morph_\textup{XY}}

\def\PsiX{\Psi_{\textup{X}}}
\def\PhiX{\Phi_{\textup{X}}}

\newcommand{\indevice}[1]{\tilde{#1}}
\newcommand{\indesign}[1]{{#1}}

\newcommand{\manifold}[1]{\mathcal{#1}}

\def\man{{\manifold{M}}}
\def\MX{{\indesign{\man}_\textup{X}}}
\def\MY{{\indesign{\man}_\textup{Y}}}
\def\Mhat{{\indevice{\man}}}
\def\MXhat{{\indevice{\man}_\textup{X}}}
\def\MYhat{{\indevice{\man}_\textup{Y}}}

\def\Pihat{{\indevice{\Pi}}} 
\def\Phihat{{\indevice{\Phi}}}
\def\Psihat{{\indevice{\Psi}}}
\def\Jhat{{\indevice{\Jvec}}}
\def\Zhat{{\indevice{Z}}}
\def\chihat{{\indevice{\chi}}}
\def\omegahat{{\indevice{\omega}}}
\def\khat{{\indevice{k}}}
\def\tHat{{\indevice{t}}}
\def\xhat{{\indevice{x}}}
\def\PiXhat{{\indevice{\pi}_\textup{X}}}
\def\PiYhat{{\indevice{\pi}_\textup{Y}}}
\def\kappahat{{\indevice{\kappa}}}

\def\ZhatX{\indevice{Z}_\textup{X}}
\def\Deltahat{\indevice{\Delta}}
\def\yhat{\indevice{y}}
\def\zhat{\indevice{z}}

\def\FX{\Ffield_{\textup{X}}}
\def\FY{\Ffield_{\textup{Y}}}

\def\Fhat{{\indevice{\Ffield}}}
\def\Ghat{{\indevice{\Gfield}}}

\def\FhatX{\Fhat_{\textup{X}}}
\def\FhatY{\Fhat_{\textup{Y}}}

\def\kappatime{\kappa_{\textup{temp}}}

\def\Plasf{\omega_{\textup P}}  
\def\Damp{\gamma}               
\def\Res{\omega_{\textup R}}    

\def\morph{\varphi}   

\def\speedofL{c}      

\def\manU{{\indesign{\man}}}
\def\manV{{\indevice{\man}}}
\def\Point{x}
\def\Morph{\morph}

\newcommand{\quotes}[1]{{\lq}{#1}{\rq}}

\def\Volts{{\textup{V}}}
\def\meters{{\textup{m}}}
\def\seconds{{\textup{s}}}
\def\Coulombs{{\textup{C}}}
\def\vect{{\textup{vec}}}

\begin{document}

\title{On Spacetime Transformation Optics:  Temporal and Spatial Dispersion}

\author{Jonathan Gratus$^{1,2}$,
Paul Kinsler$^{1,2,3}$,
Martin W. McCall$^3$,
and
Robert T. Thompson$^4$}



\address
{$^1$
  Cockcroft Institute, Keckwick Lane,
  Daresbury,
  WA4 4AD,
  United Kingdom.}

\address
{$^2$
  Physics Department,
  Lancaster University,
  Lancaster LA1 4YB,
  United Kingdom.}

\address
{$^3$
  Blackett Laboratory, Imperial College London,
  Prince Consort Road,
  London SW7 2AZ,
  United Kingdom.}

\address
{$^4$
  Department of Mathematics and Statistics,
  University of Otago,
  P.O. Box 56,
  Dunedin 9054,
  New Zealand.}

\begin{abstract}
The electromagnetic implementation of cloaking,
 the hiding of \emph{objects} from sight by diverting and reassembling
 illuminating electromagnetic fields has now been with us
 ten years,
while the notion of hiding \emph{events} is now
 five.
Both schemes as initially presented neglected the inevitable dispersion
 that arises when a designed medium replaces vacuum under transformation.
Here we
 define a transformation design protocol
 that incorporates both spacetime transformations
 and dispersive material responses in a natural and rigorous way.
We show how this methodology is applied to
 an event cloak designed to appear as a
 homogeneous and isotropic but \emph{dispersive} medium.
The consequences for spacetime transformation design
 in dispersive materials are discussed,
 and some parameter and bandwidth constraints identified.

\end{abstract}

\maketitle

\section{Introduction}
\label{ch_intro}

Transformation Design --
 the use of the mathematical transformation of reference materials
 into those interesting \quotes{device} properties
 is an area of active research interest --
 all the way from the most abstract theory and conceptualising
 \cite{McCall2015,McCall2013a,Smolyaninov-2013jo,Thompson-F-2015pra,Liu-J-2015prl},
 through to concrete theoretical proposals
 \cite{McCall2014a,MitchellThomas-MQHH-2013prl,Guenneau-AV-2012oe,Zentgraf-VTLZ-2010am}
 and technological implementations \cite{Schurig2006,Ergin-SBPW-2010s,Kadic-BSW-2015ptsa,Kinsler-TTTK-2012ejp}.

In order to achieve the graduated and controllable modulation
 of material properties that are a necessary part
 of any transformation device,
 we need to understand the underlying behaviour which generates them.
From a fundamental (microscopic) perspective,
 all non-vacuum material properties are dynamic in nature,
 resulting from the reaction of atoms, molecules, or more complex structures
 (metamaterials) to the impinging electromagnetic field,
 and thus changing how that field propagates.
It is then an effective -
 and most likely homogenised \cite{HomogMetaMat-2013pnfa,Smith-VKS-2005pre} --
 version of this dynamic process
 which we can often simplify into
 macroscopic permittivity and permeability functions,
 or perhaps even just a refractive index \cite{McCall2016}.
The sole remaining symptom of the original dynamics
 is then the frequency dependence of these constitutive quantities.

Ordinary spatial-only cloaking relies on a material response
 to achieve the device properties necessary for their operation,
 most notably one can consider the split-ring resonators
 used in the original proposal \cite{Pendry2006}.
Here,
 the intrinsically dynamic nature of material responses
 are typically not too much of a problem --
 we can specify an operating frequency and bandwidth,
 and hope that our expertise at metamaterial construction
 allows us to achieve the necessary material properties
 \cite{Milton-2010njp,Boardman-GKKLLMNRS-2011lpr,Simovski-BAK-2012am}.

Spacetime or \quotes{event} cloaks \cite{McCall2011a},
 or any other spacetime transformation devices
 \cite{McCall2015,McCall2013a,Cummer2011jo},
 are more subtle constructions,
 with an intrinsic and carefully calibrated space and time dependence.
Event cloaks have been implemented in nonlinear optics
 using dispersion as an intrinsic feature of their operation
 \cite{Fridman2012},
 in order to engineer an effective controllable speed profile.
However,
 the full details of spacetime cloaking were not addressed in
 either the original paper,
 or in the recent experiments.
A spacetime transformation not only affects
 the required material parameters,
 but also changes the underlying dynamics of the material response.
As we show in this paper,
  even the introduction of a simple one-pole
  resonance dispersion model in the design medium
  results in a number of unexpected features,
  such as the 
  resultant device's
  medium characteristics being determined
 by third derivatives of the spacetime coordinates,
 the generation of spatial dispersion from purely temporal dispersion,
 and the induction of magneto-electricity,
 beyond what is known to occur in \quotes{dispersion-free} spacetime cloaking.

This means that
 the transformation design process requires us to either
 adjust our material design to compensate
 for these extra complications,
 or engineer that extra complication so as to match our design specification.
In practise this will probably reduce to an additional
 and rather pragmatic trade-off
 of the sort we already make when attempting to build an ordinary spatial-only
 transformation device --
 what degree of approximation can we tolerate
 when attempting to match our desired performance range?
The results in this paper aim to show not only the true transformation rules
 needed for dispersive media,
 but also to inform us of how those affect the implementation trade-offs
 that will be needed.

Since electrodynamics is fundamentally a four-dimensional theory,
 all transformation optics devices
 should be seen as spacetime ones \cite{Thompson2015meta1}, 
 where
 purely spatial phase-preserving transformation devices 
 are obtained by restricting
 the time transformation to the identity.
This identity transformation, 
 when used for a cloak designed to look like vacuum,
 has the effect of forcing the ray trajectories
 to become spacelike through some regions. 
The common method for avoiding this, 
 which is to assume the cloak is constructed inside
 a background medium with refractive index sufficiently greater than one,
 is the same as assuming a non-identity time transformation.
In any case,
 for any \emph{single} frequency of operation, 
 such faster-than-light propagation
 is not strictly prohibited by relativity, 
 since the indistinguishability of phase fronts
 implies that no signal can be transmitted.
Nevertheless,
 the Kramers-Kronig relations concomitantly ensure
 that any phase-preserving cloak must be
 inherently and unavoidably dispersive.  
This unavoidable dispersion has been shown
 to betray the presence of moving cloaks \cite{Halimeh2016pra1}
 and could have other consequences.  
A fuller understanding of dispersion in transformation optics
 is therefore desirable even from the perspective
 of purely spatial applications, 
 as well as from the potential for dispersion engineering
 \cite{Moccia2016opt}.

This paper is organised as follows.
First,
 in section \ref{S-background} we describe a
 dispersive spacetime cloak
 in general terms and introduce
 our mathematical \quotes{morphism} terminology.
Section \ref{S-onepole} then describes the handling of the one-pole
 or Lorentzian resonance which will make up our design medium --
 that which we want observers to infer exists.
We show how this oscillator,
 whose most direct description is its temporal differential equation,
 can be represented in either the frequency domain
 or in terms of an integral kernel.
Next,
 section \ref{S-tdesign} shows how the transformation design process
 is implemented in this most general dispersive and spacetime case,
 either as operators,
 or again as an integral kernel.
Section \ref{S-transformedpole} then applies the general machinery
 to our chosen case,
 and calculates what kind of temporally and spatially varying oscillator
 is needed for a spacetime cloak to perfectly mimic
 our desired \quotes{one-pole} appearance.
We then use these results to estimate some bandwidth/parameter constraints
 that could be applied to our cloak design if we only had
 temporal oscillators to build with.
Finally,
 in section \ref{S-conclusion} we present our conclusions.

%
\section{Background}\label{S-background}

\begin{figure}[ht]
\centering{
\includegraphics[width=0.79\columnwidth]{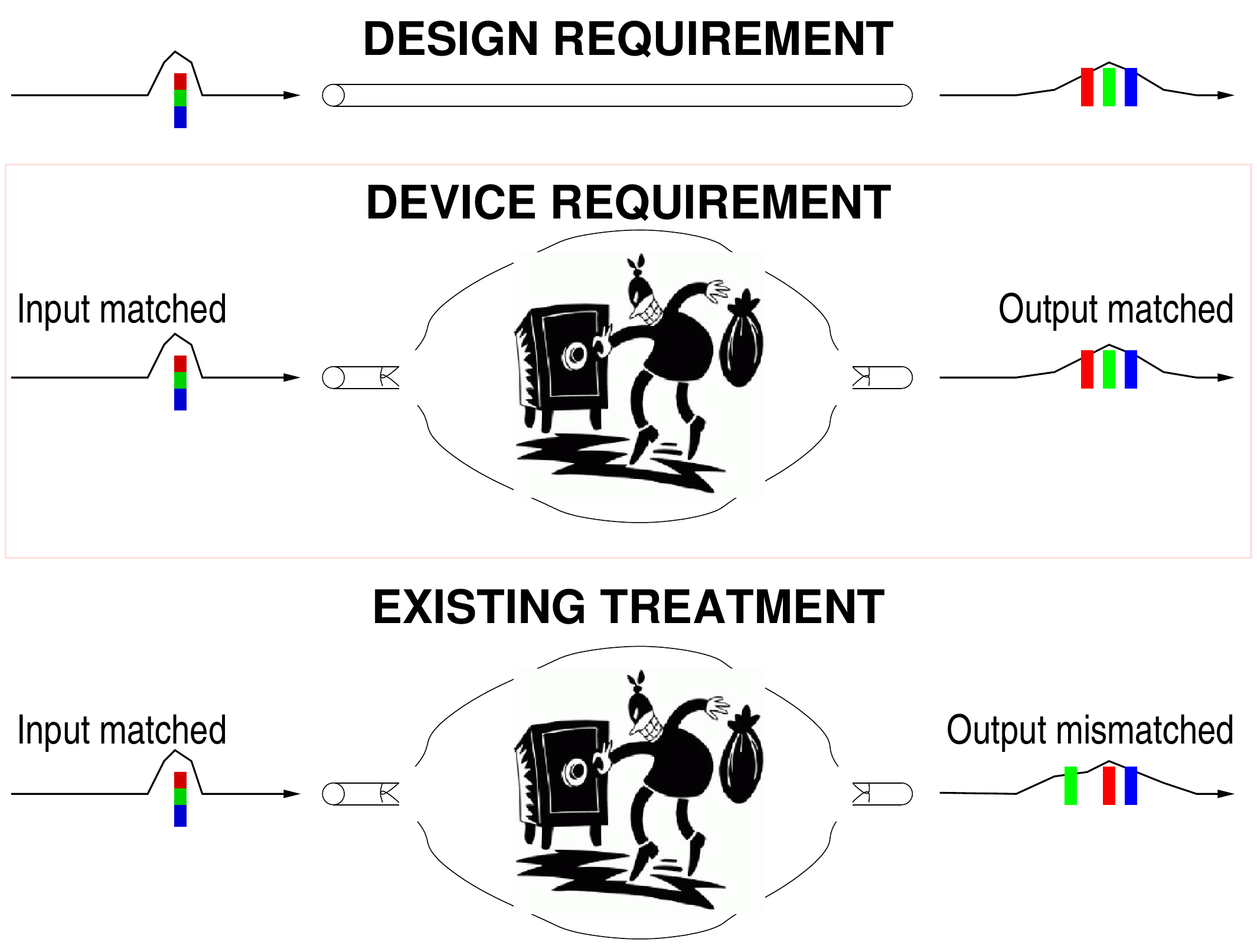}
}
\caption{
The requirements of spacetime transformation schemes:
 whilst the visible behaviour of the device
 should only be to (top) alter any incident illumination by
 the expected dispersion properties,
 the actual transformation device (middle)
 also must hide a chosen event from any observer.
Existing treatments,
 which ignore the effects of spacetime transformations
 on the dispersion properties,
 will not perfectly match the design requirement --
 even if all the beam steering and scattering suppression
 is still implemented correctly. 
The implication is that the observer will suspect
 that tampering has occurred --
 depicted here by the output illumination pulse 
 having the wrong frequency chirp -- 
 even though the event itself is still perfectly hidden.
}
\label{fig-diagram}
\end{figure}

In this work we envisage a simple event cloak scheme,
 but in contrast to the original formulation based on
 a homogeneous, isotropic, and dispersion-free background,
 we want our \emph{device} to hide an event inside
 a dispersive medium.
That is,
 we are going to \emph{design} an event cloak so that
 the device itself,
 despite its many complications,
 seems to an observer to be acting like
 a simple homogeneous and isotropic material
 that follows the standard Lorentz model.
Naturally,
 since this is a linear system,
 the method could be straightforwardly generalised
 to encompass a sum of Lorentz oscillators as well
 \cite{Dirdal-S-2013pra};
 and it is worth noting that with careful parameter choice,
 the Drude model for material response can be encoded
 within the Lorentz model.
We also show how to define material responses
 as differential equations for polarization 2-forms,
 and this methodology is general enough to also handle
 many other (i.e. non-Lorentzian) response models.
We also provide an integral kernel approach
 that is even more general.

Now consider how our cloak needs to work in practise.  An optical
pulse which started with a fixed phase, but then travelled through an
ordinary dispersive medium, will typically emerge with some chirp,
simply because its different frequency components experience different
phase velocities, as well as generating a group velocity for the pulse
as a whole.  This situation is depicted at the top of
fig. \ref{fig-diagram}, and is how our spacetime cloaking device is
designed to appear to an observer.  We however, want to hide an event
inside a different \quotes{device} medium, whose spatial and temporal
properties not only hide our chosen event, but \emph{also} mimic an
ordinary dispersive medium of our choice, as seen at the middle of
fig. \ref{fig-diagram}.  If we do not properly consider \emph{all}
aspects of how the medium might need to be adapted to the true
spacetime nature of our cloaking transformation, an observer may see
evidence of tampering despite the cloaked event remaining hidden,
as shown at the bottom of fig. \ref{fig-diagram}.

To be clear,
 this chirp-induced betrayal of the cloak's presence is an entirely different effect
 than that of wave scattering by a reduced parameter cloak \cite{Colquitt-JMMBM-2013prsa},
 and can be present even when ray trajectories are perfectly preserved by the cloak. 
Our interest in this work is not on the effectiveness or scattering reduction
 achieved by some implementation of a cloaking device. 
Although an important point,
 and certainly so in the more mature area of spatial cloaking,
 its applicability is more relevant to specific implementations
 than to the fundamentals we address here.

The goal is therefore to construct constitutive relations
 of a medium such that
 (a) there is a space time cloak,
 and
 (b) that the observer sees frequency dispersion corresponding
 to a one pole Lorentz resonance.
The challenge with incorporating dispersion
 is that the new constitutive relations are now functions
 of both frequency $\omega$ and time $t$.
However,
 since $\omega$ and $t$ are conjugate Fourier variables,
 we need somehow to give meaning to constitutive media properties --
 the permittivity and permeability --
 which will depend on both.
Further,
 since we perform transformations in space and time,
 the new constitutive relations will mix
 both time $t$ and space $x$
 with  frequency $\omega$ and wavevector $k$.
Thus the new constitutive relations are both inhomogeneous in space
 \emph{and} spatially dispersive \footnote{As
  an interesting aside,
  in this there is an analogy to quantum mechanics,
  in that when writing down expressions in $t$ and $\omega$
  they do not commute and we find that the commutator $[t,\omega]=-i$.
 This implies that a solution to Maxwell's equations
  cannot be single mode $e^{i\omega t}$.
 This may give some insight into the nature of quantum mechanics
 by analogy with electromagnetic modes in dispersive inhomogeneous media.}
 \cite{Gratus-M-2015jo,Gratus-T-2011jmp,agranovich1985crystal}.

In this work we will consider the case of general spacetime
 transformation design incorporating dispersive effects,
 for which we will typically use
 the idea of an event cloak as a proxy.
We will use two approaches
 to describing the constitutive properties:
 a differential operator approach
 which is particularly useful for the one-pole resonance
 that is the main focus of this article;
 and an integral kernel method valid for more general linear media.
Our mathematical underpinning of the physics is now as follows.

The device itself will consist of some complicated arrangement
 of material properties,
 but it is designed to appear as if it were simple.
The mathematical description of this (simple) design space
 takes place on a \quotes{design manifold} $\man$,
 and the description of the device
 takes place on a \quotes{device manifold} $\Mhat$.
As depicted in fig. \ref{fig-morphism},
 linking the two is a transformation or morphism $\morph$
 which expresses how spacetime points on $\Mhat$
 (i.e. inside the device)
 need to be located on $\man$ --
 so that fields travelling through the device emerge in time and space
 as if having travelled through our designed $\man$.
For example,
 in the well-known spatial and dispersionless cloak case,
 we have that origin-avoiding trajectories in $\Mhat$
 become straight lines in $\man$;
 for a dispersionless event cloak
 the curtain-map 
 (\cite{McCall2011a}, Fig. 6)
 used
 converts between a space with a diamond-shaped cloaking region
 and one with the diamond closed up.

\begin{figure}
  \begin{center}
    \resizebox{0.80\columnwidth}{!}{
    \input{fig02_morphismpktx_pdf_t}}
  \end{center}
    \caption[] { \label{fig-morphism} Diffeomorphism:  points
      $\Point$ in the device manifold $\manV$ are mapped to points
      ${\morph} (\Point)$ in the design manifold $\manU$ by the
      morphism (mapping) ${\Morph}$.  Even though throughout this work
      we use, where possible, a coordinate independent notation,
      coordinates can be constructed. For example, with
      $\indevice{\xi}$
      mapping manifold
      $\manV$ to $\Real^4$ and $\xi$ mapping $\manU$ likewise.
      Nevertheless, at the
      implementation stage, specific coordinate systems are
      invaluable, since $\morph$ can be represented by $\xi \circ
      {\morph} \circ \indevice{\xi}^{-1}$ in $\Real^4$ which then
      gives us the \quotes{blueprint} for our device.  }
\end{figure}

In this paper we will use the coordinate free notation
 of exterior differential forms\footnote{
  We have used the standard tools of coordinate free differential
  geometry, i.e. the wedge product, exterior derivative, internal
  contraction, Lie derivative, Hodge dual and pullback. 
These are defined as follows:\\
~
The exterior derivative $d$ increases the degree of a form by 1. For
0--forms also known as scalar fields, $d\phi=\pfrac{\phi}{x^a} ~ dx^a$
and on the wedge product $\alpha\wedge
  \beta$ of a $p$--form $\alpha$, and a form of arbitrary degree
  $\beta$, via $d\left(\alpha\wedge \beta\right) = d\alpha\wedge
  \beta +(-1)^p\alpha\wedge d \beta$.\\
~
The internal contraction is defined for
  the vector field $V$, $i_V$ acts on a 1--form $\alpha$ as $i_V\alpha =
  V^\alpha \alpha_\alpha$, and on the wedge product $\alpha\wedge
  \beta$ of a $p$--form $\alpha$, and a form of arbitrary degree
  $\beta$, via $i_V\left(\alpha\wedge \beta\right) = i_V\alpha\wedge
  \beta +(-1)^p\alpha\wedge i_V \beta$. The internal contraction
  operator $i_V$ therefore reduces the degree of a form by 1.\\
~
The Lie derivative maps $p$--forms to $p$--forms
 via Cartan's identity $L_V=d ~ i_V+i_V ~ d$.\\
~
The Hodge dual $\Star$ takes $p$ forms to $4-p$ forms,
 can be succinctly and uniquely defined by the requirement
 that it is tensorial,
 $\Star(\alpha\wedge dx^a)= g^{ab} i_{\partial/\partial x^a}\Star \alpha$
 and that the 4-volume form $\Star 1$
 has the correct orientation with $\Star\Star 1 = - 1$.\\
~
Given a map $\morph:\Mhat\to \man$, the pullback $\morph^\pstar$ maps
$p$--forms on $\man$ to $p$--forms on $\Mhat$ and satisfies: For 0--forms
$\morph^\pstar\phi=\phi\circ\morph$ and on arbitrary forms $\alpha$,
$\beta$,
$\morph^\pstar(\alpha\wedge\beta)
=\morph^\pstar(\alpha)\wedge\morph^\pstar(\beta)$
and $\morph^\pstar(d \alpha)=d (\morph^\pstar\alpha)$.\\
~
 All these operations distribute across addition.}
\cite{Flanders1963,HehlObukhov} --
 although we could,
 for example,
 always use an indexed notation,
 this would complicate the equations unnecessarily.  
When we map the manifolds onto charts with coordinate systems
 (see fig. \ref{fig-morphism}),
 we are using an \quotes{active transformation}, 
 not a \emph{passive} \quotes{coordinate transformation}.
It is important to appreciate that the physics of Transformation Optics
   is independent of this coordinate representation,
 a point discussed   further in \cite{Thompson-F-2015pra}.
One other crucial point is
 that on both manifolds, the underlying spacetime metric is taken to
 be Minkowski with Lorentzian signature $(-,+,+,+)$.  Much of the
 work here applies for general constitutive relations on curved
 spacetimes, however we do exploit the fact that the there is a
 timelike killing vector given by $\partial_t$.  
Further,
 although the notion of an effective \quotes{optical metric}
 \cite{leonhardt2012geometry} can indeed be useful, 
 we do not need or use it here.
However an alternative interpretation of transformation optics
 is to consider two metrics on the device manifold $\indevice{\man}$: 
 the Minkowski $\indevice{g}$ and the optical $g_{opt}$. 
Since there is no longer a single preferred metric,
 knowing which aspects of electromagnetism
 are independent of the metric is useful, 
 as in the premetric formulation \cite{HehlObukhov}.

Maxwell's equations in our chosen notation are
\begin{align}
  d\Ffield = 0 
\qquadand
  d \Star \left( \epsilon_0 \Ffield + \Pi \right)
=
  \speedofL^{-2} \Star \Jvec~,
\label{intro_Max_FG}
\end{align}
where $\Star$ denotes the Hodge dual,
 $\speedofL$ is the speed of light in vacuum, 
 $\Jvec$ is the current density, 
 and $\epsilon_0$ is the vacuum permittivity\footnote{Note
   that dimensionally $[dx^a]=\meters$ and $[\partial_a]=\meters^{-1}$.
  The metric, which converts vectors to 1--forms, has
   units of $\meters^2$. The Hodge dual has units which depend on the
   degree: $[\Star\alpha]=\meters^{4-2\deg(\alpha)} [\alpha]$. 
  Integration and exterior differentiation have no effect on units. 
  The components of the electric field have dimensions
   $[\Efield^a]=\Volts\meters^{-1}$. 
  Thus the vector $\Efield_\vect=\Efield^a\partial_a$
   has dimensions $[\Efield_\vect]=\Volts\meters^{-2}$, 
   and the dimensions of the 1--form $E$ are $[\Efield]=\Volts$
   (and $[\int \Efield]=\Volts$). 
 Likewise,
  since $[\Efield]=[\speedofL\Bfield]$,
  then $[\Bfield]=\Volts\seconds\meters^{-1}$. 
 Also,
  $[\Ffield]=[ dt \wedge \Efield]
    =[ \speedofL \, \Star ( dt \wedge \Bfield ) ]
    =[ \speedofL \, ( dt \wedge \Bfield ) ]
    =\seconds\Volts$. 
 Thus $[\epsilon_0 \Ffield]
    = (\Coulombs\Volts^{-1}\meters^{-1})(\seconds\Volts)
    = \Coulombs\meters^{-1}\seconds$. 
 This gives $[\Star \epsilon_0 \Ffield]
    =\Coulombs\meters^{-1}\seconds$
  and hence $[\epsilon_0 \speedofL^2 d\Star \Ffield]
    =\Coulombs\meters\seconds^{-1}$. 
 The components of the current density $\Jvec^a$
  have dimensions $[\Jvec^a]=\Coulombs\seconds^{-1}\meters^{-2}$. 
 Thus the vector current is
  $[\Jvec_\vect]=\Coulombs\seconds^{-1}\meters^{-3}$,
  then 1--form $[\Jvec]=\Coulombs\seconds^{-1}\meters^{-1}$
  and the 3--form $[\Star \Jvec]=\Coulombs\seconds^{-1}\meters$. 
 Thus \eqref{intro_Max_FG} is dimensionally correct.}.  
We can also split the electromagnetic 2--form $\Ffield$ 
 so that its electric field $\Efield$ and magnetic field $\Bfield$ 1--form 
 sub-components are visible, and show how the dielectric
 polarization $\Pol$ and the magnetisation $\Mag$
 appears inside the polarization 2--form $\Pi$, i.e.
\begin{align}
  \Ffield = dt\wedge \Efield + \speedofL \Star \left( dt\wedge \Bfield \right)
\qquadand
  \Pi = dt\wedge \Pol - \Star \left( dt \wedge \Mag\right)
.
\label{intro_FG}
\end{align}
where
\begin{align}
\Efield = i_t \Ffield
,\quad
\Bfield = - \speedofL^{-1} \,i_t \Star \Ffield
,\quad
\Pol = i_t \Pi
\quadand
\Mag = i_t\Star\Pi~.
\label{intro_EBDH}
\end{align}
Here $i_t = i_{\pfrac{}{t}}$ is the
internal contraction operator taking the 2--form $dt \wedge \Efield$ to the
1--form $\Efield$, for example. 
We use the symbol $\partial_t=\pfrac{}{t}$ etc. 
 for both the partial derivative and the corresponding vector field.

Since we will have to deal with the dynamics
 of the material response explicitly,
 we do not use a susceptibility tensor  $\chi$ as
 might normally be expected.
Instead we represent the material using
 constitutive operators $\Psi$ and $\Phi$
 which take 2--form fields to 2--form fields,
 so that
\begin{align}
  \Psi
 ~
  \Pi
=
  \Phi
 ~
  \Ffield
,
\label{intro_constitutive_oper_unhat}
\end{align}
or, in components
\begin{align}
 \Psi^{cd}_{ab} ~ \Pi_{cd}
=
 \Phi^{cd}_{ab}  ~ \Ffield_{cd} ~,
\label{intro_constitutive_Phi_Psi_unhat}
\end{align}
where $\Psi^{cd}_{ab}$ and $\Phi^{cd}_{ab}$ are antisymmetric in $cd$
and $ab$.
The reasons for using \eqref{intro_constitutive_oper_unhat}
 will become clear later,
 but for now we note that the structure of this
 matches that of the differential equation for $P$
 as given below in \eqref{intro_pde_P}.
The summation convention is used throughout,
 and spacetime indices $a,b, \ldots =0,1,2,3$,
 space indices $\mu,\nu, \ldots = 1,2,3$.
Finally,
 throughout this paper we use the non-unitary,
 angular frequency definition of Fourier transform.

%
\section{Design goal: the one-pole resonance}\label{S-onepole}

Our goal requires that the constitutive properties of the
 design medium appear to be a stationary,
 homogeneous, isotropic material
 with the behaviour of a single-pole Lorentz oscillator.
In the simple case of a global Lorentz coordinate system,
 we could simply write a temporal differential equation
 (see e.g. \cite{Kinsler-2011ejp})
 using partial derivatives acting on the relevant components;
 however our aim
 demands that we use a more general spacetime form with Lie derivatives,
 i.e.
\begin{align}
\left( L_t^2 + \Damp L_t + \Res^2\right) \Pol = \Plasf^2  ~ \Efield~,
\label{intro_pde_P}
\end{align}
where $L_t$ denotes the Lie derivative $L_t = L_{\partial_{t}}$ and
$\Res$ is the resonance frequency, $\Damp$ is the damping, $\Plasf$ is
the coupling strength. We choose, as
a matter of model construction, Lie transport, and therefore the Lie
derivative, as opposed to parallel transport and its corresponding
covariant derivative. For static media the two models are
indistinguishable. However, the spacetime transformation optics for a covariant
derivative formulation will be distinct.

In the simple case of a global Lorentz coordinate system,
 the Lie derivatives can be replaced by partial derivatives
 acting on the relevant components.
In such a situation,
 the frequency domain behaviour of the Lorentz oscillator has the form
\begin{align}
  \FTepsilon(\omega)
=
  \epsilon_0
  \left(
    1
   +
    \frac{\Plasf^2}{-\omega^2+i\Damp\omega+\Res^2}
  \right)
\qquadand
\mu=\mu_0~,
\label{intro_one_pole_resonance}
\end{align}
where $\mu_0$ is the vacuum permeability.
The electric and displacement fields $\Efield, \Dfield$ are linked
 by constitutive relations depending on the material polarization $\Pol$
 which can be extracted from either \eqref{intro_pde_P} or
 \eqref{intro_one_pole_resonance} above.
Assuming that the magnetic response is that of the vacuum, i.e. $\Mag=0$,
 we have
\begin{align}
\Hfield = \mu_0^{-1} \Bfield
\qquadand
\Dfield = \epsilon_0 \Efield + \Pol~.
\label{intro_const_rel}
\end{align}

One point of note is that
 the operator $\left( L_t^2 + \Damp L_t + \Res^2 \right)$
is not one-to-one.
This is because
 $\left( L_t^2 + \Damp L_t + \Res^2 \right) e^{i\sigma t} =0$
 where $\sigma$ is a root of ${-\omega^2+i\Damp\omega+\Res^2}=0$.
In this case the 1--form $C$ such that $\Pol= C e^{i\sigma t}$
 satisfies $L_t C=0$.
As a result,
 we find that ${\epsilon}$ is infinite,
 so we denote these `plasma resonance modes',
 and exclude them from our analysis.
However,
 if the damping $\Damp>0$
 and we deal only with real frequencies,
 then the plasma resonance modes are automatically excluded.

We now need to represent this dynamic material response
 as a constitutive property of Maxwell's equations.
We will do this first in a frequency domain picture,
 then as an integral kernel.

%
\subsection{Operator representation}\label{S-frequencydomain}

When relating the time and frequency versions
 of the dynamic material response,
 it is useful to first
 show explicitly how the constitutive relations
 may be written as in
  \eqref{intro_constitutive_oper_unhat} 
 for the case of the one-pole resonance considered here. We do this by
 proposing (and proving) the following Lemma.

\begin{lemma}
We can write the material response from
 \eqref{intro_pde_P}, \eqref{intro_const_rel} as the
 constitutive relations in \eqref{intro_constitutive_oper_unhat}
 by setting
\begin{align}
\Psi=L_t^2 + \Damp L_t + \Res^2
\qquadand
\Phi=\Plasf^2 ~ dt\wedge i_t~,
\label{intro_constitutive_Psi_Phi_uhat}
\end{align}
i.e.
\begin{align}
  \left( L_t^2 + \Damp L_t + \Res^2 \right) 
  \Pi
=
  \Plasf^2 ~ dt\wedge i_t \Ffield~,
\label{intro_constitutive_uhat}
\end{align}
as long as we avoid plasma resonance modes.
\end{lemma}
\begin{proof}
We can write \eqref{intro_const_rel} and \eqref{intro_pde_P} as
\begin{align}
\left( L_t^2 + \Damp L_t + \Res^2 \right) i_t \Pi = \Plasf^2 ~ i_t \Ffield
\qquadand
i_t \Star \Pi = 0~,
\label{intro_constitutive_PDE_uhat}
\end{align}
However, from \eqref{intro_FG} we have
$\Pi=dt\wedge i_t \Pi$.
Then from \eqref{intro_pde_P} we have
\begin{align*}
\left( L_t^2 + \Damp L_t + \Res^2 \right) \Pi
=
\left( L_t^2 + \Damp L_t + \Res^2 \right)
\left( dt \wedge i_t \Pi \right)
=
dt \wedge \left( L_t^2 + \Damp L_t + \Res^2 \right) \left( i_t \Pi \right)
=
  \Plasf^2 
  dt \wedge i_t 
  \Ffield
.
\end{align*}
Thus \eqref{intro_constitutive_PDE_uhat} implies
\eqref{intro_constitutive_uhat}. Clearly
\eqref{intro_constitutive_uhat} implies the first equation in
\eqref{intro_constitutive_PDE_uhat}. In addition since $\partial_t$ is
Killing so that $L_t \Star =\Star L_t$ we have
\begin{align*}
\left( L_t^2 + \Damp L_t + \Res^2 \right) i_t \Star  \Pi
=
  i_t \Star  \left( L_t^2 + \Damp L_t + \Res^2 \right) 
  ~
  \Pi
=
  i_t \Star  
  \left(
    \Plasf^2
    dt \wedge i_t
    \Ffield 
  \right)
=0~.
\end{align*}
As discussed above,
 we have excluded any consideration of  plasma resonance modes.
As a result,
 the proof is completed since the last equation above
 implies the second equation in \eqref{intro_constitutive_PDE_uhat}.
\end{proof}

Following this Lemma,
 we can create frequency dependent constitutive relations
 which can be specified by the components of $\Psi$ and $\Phi$.
By replacing $i\omega\leftrightarrow L_t$
 we get
\begin{align}
\Psi^{cd}_{ab}(\omega)
=
({-\omega^2+i\Damp\omega+\Res^2})
(\delta^c_a\delta^d_b-\delta^d_a\delta^c_b)
,
\quad
\Phi^{0\mu}_{0\nu}={\Plasf}^2 ~ \delta^\mu_\nu
,
\quad
\Phi^{\mu\nu}_{\sigma\rho}(\omega,k,t,x)=0~.
\label{intro_constitutive_Phi_uhat_alt}
\end{align}
This means that the usual susceptibility matrix,
 in Fourier transform space,
 is defined simply by
\begin{align}
\chi^{cd}_{ab}(\omega)=(\Psi^{-1})^{cd}_{ef}(\omega)\Phi^{ef}_{ab}(\omega)~.
\label{intro_chi}
\end{align}
We prefer to work with $\Psi$ and $\Phi$ instead
 of the more usual $\chi$,
 as we can then avoid the difficulties in forming the operator inverse
 of $\Psi$ (cf.  \eqref{intro_constitutive_Phi_uhat_alt}), 
 i.e. where we need that
\begin{align}
  \tfrac{1}{2}
  \left(\Psi^{-1}\right)^{cd}_{ef} 
  \Psi^{ab}_{cd}
&=
  \left( 
    \delta^a_e \delta^b_f - \delta^b_e \delta^a_f 
  \right)
.
\end{align}

%
\subsection{Integral kernel representation}\label{S-integralkernel}

\begin{figure}
\centering{
\setlength{\unitlength}{0.03\textwidth}
\begin{picture}(12,10)
\put(0,1){\vector(0,1){4}}
\put(0,1){\vector(4,1){4}}
\put(0,1){\vector(4,-1){4}}
\put(-0.9,3){$y^0$}
\put(4,1.5){$y^1$}
\put(4,0){$y^2,y^3$}
\put(3,0){
   \includegraphics[width=0.3\textwidth]{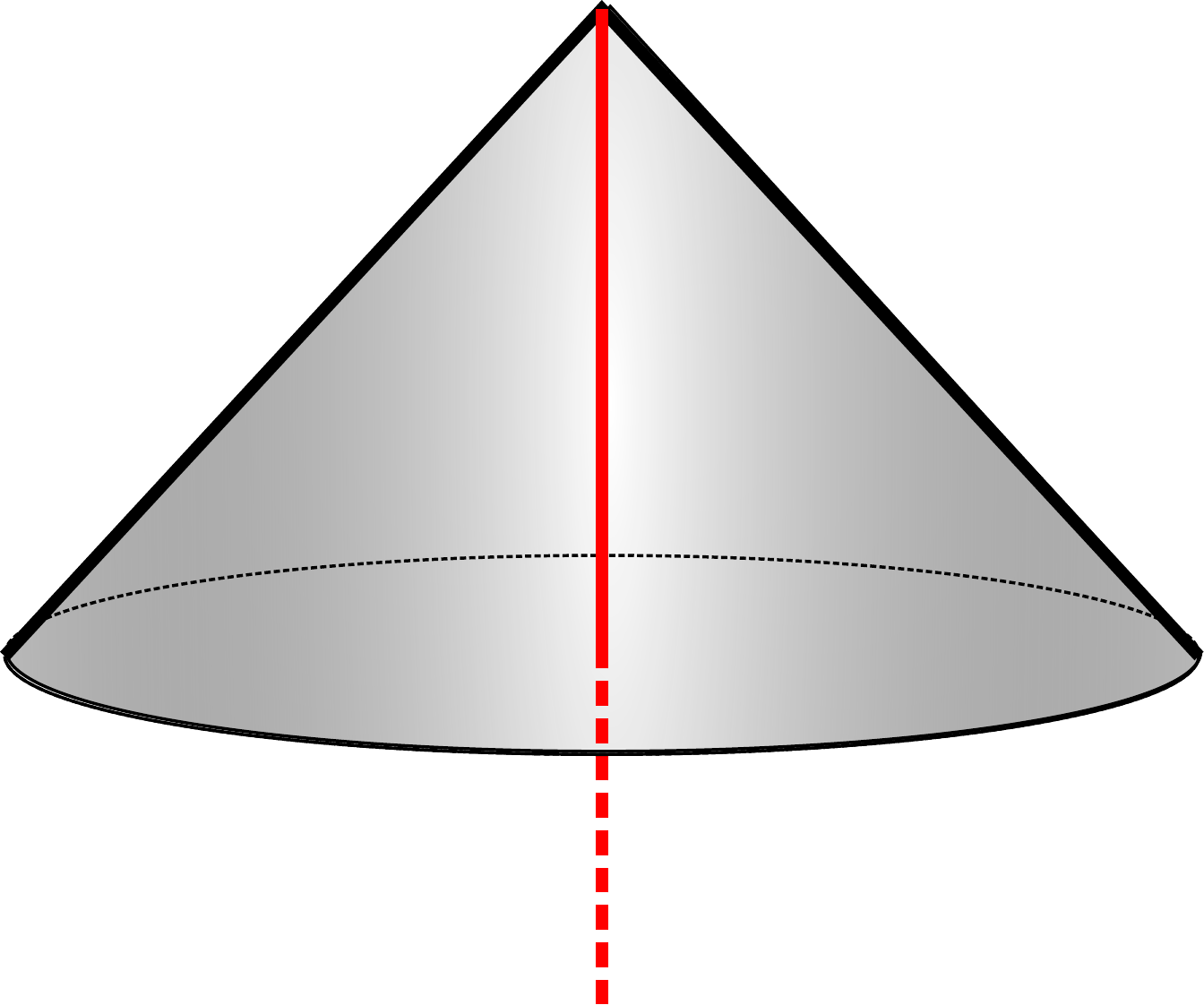}}
\put(7,8.6){$(x^0,\Vx)$}
\end{picture}
} \caption{The support (red) in $\MY$ of the kernel $\kappa$ given in
  \eqref{intro_ker_kappa}, for the point $(x^0,\Vx)\in\MX$ lying in
  its backward lightcone. This may be contrasted with the support of
  the transformed kernel, given below in fig. \ref{fig_kernel_support_real}}
\label{fig_kernel_support_virtual}
\end{figure}

It is well-known that the polarization properties of a linear, isotropic, non-magnetic medium with a local, but non-instantaneous response may be written in terms of an integral kernel as \cite{Post1997}
\begin{equation}\label{epstminustau}
P(t)=\int_{-\infty}^{\infty}\kappatime(t,\tau)E(\tau)d\tau~,
\end{equation}
where causality requires that $\kappatime(t,\tau)=0$ for $t>\tau$. In
a usual time translation invariant kernel
$\kappatime(t,\tau)=\kappatime(t-\tau)$ and Fourier techniques are
convenient. However, a necessary feature of a spacetime
transformation is that the resulting kernel will cease to be time
translation invariant, as seen in \eqref{OnePKer_phistar_kappa} and
\eqref{OnePKer_phistar_kappa_x_dx} below.

We seek to generalise the integral kernel approach of
\eqref{epstminustau} to four dimensions so that both temporal and
spatial effects may be accounted for simultaneously. This will of
course be necessary as a precursor to understanding how the
constitutive relations behave under the action of a spacetime cloak.

Both the electromagnetic field $\Ffield$ and the polarization field
$\Pi$ are 2--forms. We write
\begin{align}
  \Pi_{ab}(x) 
=
  \tfrac{1}{4}
  \epsilon^{cdef}
  \int_{\man} 
    \kappa_{abcd}(x,y)
    \Ffield_{ef}(y)
    dy^0\wedge dy^1\wedge dy^2\wedge dy^3
.
\label{intro_ker_Pi}
\end{align}
where $\kappa_{abcd}(x,y)$ is a generalisation of
$\kappatime(t,\tau)$. The two parameters $t$ and $\tau$ are
generalised to two spacetime events $x$ and $y$. For convenience let
$x\in\MX$ and $y\in\MY$ where $\MX$, $\MY$ are copies of the
spacetime $\man$. We use coordinates $x=(x^0,x^1,x^2,x^3)$ with $x^0=t$
and $\Vx=(x^1,x^2,x^3)$.

As in standard Green's function theory we lift the linear operators
on $\man$, $\Psi$ and $\Phi$ given by \eqref{intro_constitutive_Psi_Phi_uhat}
into linear operators $\PsiX$ and $\PhiX$
on the product $\MX\times\MY$ by requiring they
act only on the $x$ coordinates. That is we set
\begin{align}
\PsiX = L_{x^0}^2 + \Damp L_{x^0} + \Res^2
\qquadand
\PhiX = \Plasf^2 dx^0 \wedge i_{x^0}~.
\label{Kernel_PsiX_PhiX}
\end{align}
Thus $\Pi$ given by \eqref{intro_ker_Pi} will satisfy the differential
equation \eqref{intro_constitutive_uhat} if the 4--form $\kappa$ satisfies
\begin{align}
  \PsiX 
  \kappa 
=
  \PhiX
  \Delta
,
\label{intro_ker_green}
\end{align}
where
\begin{align}
  \Delta
=
  \tfrac{1}{4}
  \epsilon_{abcd} 
  \delta^{(4)}(x-y) 
  ~
  dx^{a}
  \wedge dx^{b}
  \wedge dy^{c}
  \wedge dy^{d}
,
\label{intro_ker_Delta}
\end{align}
and
\begin{equation}
  \kappa 
=
  \tfrac{1}{4}
  \kappa_{abcd}(x,y)
  ~
  dx^a
  \wedge dx^b
  \wedge dy^c
  \wedge dy^d
,
\label{intro_kernal}
\end{equation}
are 4--forms on the product of two copies of
spacetime, $\MX\times\MY$.

We may write a solution of \eqref{intro_ker_green} as
\begin{align}
  \kappa
=
  \Plasf^2
  \frac{\theta(x^0-y^0)\delta^{(3)}(\Vx-\Vy)}
       {2(\sigma_+-\sigma_-)}
  \left[ 
    e^{\sigma_+(x^0-y^0)} - e^{\sigma_-(x^0-y^0)} 
  \right]
  ~
  \epsilon_{\lambda\nu\rho}\
  ~
  dx^0 \wedge dx^\lambda \wedge dy^\nu \wedge dy^\rho
,
\label{intro_ker_kappa}
\end{align}
where $\theta(x^0-y^0)$ is the Heaviside function, 
 $\Vx=(x^1,x^2,x^3)$
 and $\sigma_+$, $\sigma_-$
 are the two roots
\begin{align}
\sigma_{\pm}=-\tfrac12\Damp\pm\sqrt{\tfrac14\Damp^2-\Res^2}~.
\label{intro_ker_mu_roots}
\end{align}
The proof that \eqref{intro_ker_kappa} is a solution to
\eqref{intro_ker_green} is given in lemma \ref{lm_ker_soln} in the
Appendix. The support of the kernel $\kappa$, that is the set of
points $(x,y)$ such that $\kappa(x,y)\ne0$,
 is given in fig. \ref{fig_kernel_support_virtual}. 
Indeed the support of any causal
 temporally dispersive kernel is given by (a subset of)
 fig. \ref{fig_kernel_support_virtual}.

We see that the coefficients in $\kappa$ in \eqref{intro_ker_kappa}
 are actually functions of the difference $x-y$, 
 using the affine structure of Minkowski spacetime. 
This is because it is the Green's function for a
 spacetime-homogeneous differential equation,
 i.e. the differential operator is independent of position.  
However in general,
 either for general relativity or for inhomogeneous media,
 $\kappa$ will \emph{not} have this structure. 
Thus we will see that $\kappahat$,
 which is the kernel to generate $\Pihat$,
 the spacetime deformation of $\Pi$,
 will not correspond to any difference.

%
\section{Transformation Design}\label{S-transformationdesign}\label{S-tdesign}

As motivated in section \ref{S-background},
 we are now going to invoke some suitable diffeomorphism $\morph$
 as a part of a transformation design process.
At first,
 the steps that need to be taken might seem relatively straightforward.
Under the diffeomorphism $\morph: {\Mhat} \rightarrow \man$
 the operators $\Psi$ and $\Phi$ become $\Psihat$ and $\Phihat$ respectively.

Although the goal for this article is to find the explicit form of the
operators $\Psihat$ and $\Phihat$  for the diffeomorphism corresponding to a spacetime cloak with a single pole
resonance, in fact all the results in this section apply to any
constitutive relations which can be represented by the differential
equation \eqref{intro_constitutive_oper_unhat}, for example a sum of
Lorentz operators. In this general case
the transformed polarisation 2--form $\Pihat$, satisfies the operator
equation
\begin{align}
  \Psihat
 ~
  \Pihat
=
  \Phihat
 ~
  \Fhat
.
\label{intro_constitutive_oper_hat}
\end{align}
Here $\Psihat$ and $\Phihat$ will be constructed out of the (morphed)
Lie derivatives $L_\tHat$ and $L_\xhat$ as well as internal
contractions.  Of course we will again stay away from resonances so
that $\Psihat ~ \Pihat=0$ if and only if $\Pihat=0$.  It is now
tempting to immediately apply the identities $-i\omegahat\leftrightarrow
L_\tHat$ and $i\khat\leftrightarrow L_\xhat$ in order to obtain matrix
entries $\Psihat^{cd}_{ab}(\omegahat,\khat,\tHat,\xhat)$ and
$\Phihat^{cd}_{ab}(\omegahat,\khat,\tHat,\xhat)$ and then apply the inverse of $\Psihat^{cd}_{ab}$ in order to
construct a single susceptibility matrix
$\chihat^{cd}_{ab}(\omegahat,\khat,\tHat,\xhat)$. However,
we must be careful when we write down
$\Psihat(\omegahat,\khat,\tHat,\xhat)$ as neither $\omegahat$ and $\tHat$, nor $\khat$ and $\xhat$, commute
--
 just as happens in a similar manner in quantum mechanics.
Therefore, in writing down
 $\Psihat(\omegahat,\khat,\tHat,\xhat)$
 we must be sure to retain the proper ordering of these variables;
 and this consideration makes computation of the inverse
 even more problematic.
Thus although we might \emph{formally} write down $\Psihat^{-1}$,
 calculating it in practise would be very difficult.

It is helpful to identify two distinct scales 
 of spatio-temporal variation:
 the scale on which the spacetime cloak varies,
 and the scale on which the envelope of the optical field varies. 
If the spacetime cloak varies sufficiently slowly
 that we can neglect second order and higher derivatives of $\morph$,
 then we refer to this situation as the
 \emph{Gradual Transformation Approximation} (GTA). 
A similar approximation, 
 albeit more restrictive, 
 applied to the envelope of the optical field
 results in the so-called \emph{Slowly Varying Envelope Approximation} (SVEA),
 i.e. the fields are of the form
\begin{align}
  \Efield(\tHat,\xhat)
&=
  \mathcal{A}(\tHat,\xhat)\exp\left(-i\omegahat_0 \tHat+i\khat_0 \xhat\right)
,
\label{eqn-fieldenvelope}
\end{align}
 where 
 $\mathcal{A}(\tHat,\xhat)$ is a pulse envelope
 whose variation can be considered negligible \cite{Boyd-NLO}. 
Assuming the SVEA,
 and writing the Lie derivatives $L_{\tHat}$ and $L_{\xhat}$
 to the right of $\tHat$ and $\xhat$,
 then the role of $\omegahat_0$ and $\khat_0$
 can be identified with $\omegahat$ and $\khat$
 in $\chihat^{cd}_{ab}(\omegahat,\khat,\tHat,\xhat)$. 
If we assume both the GTA and the SVEA, 
 the transformed constitutive relations are much simpler,
 but still necessarily contain both spatial and temporal dispersion,
 as given in \eqref{Appr_chihat_elec} below.

Now let $\morph$ be the diffeomorphism
 between the device manifold $\Mhat$ and the design manifold $\man$,
\begin{align}
\morph:\Mhat\to \man~.
\label{gen_phi}
\end{align}
See fig. \ref{fig-morphism}.  
Let $\Starhat$ be the Hodge dual
defined with respect to the Minkowski spacetime metric on $\Mhat$. 
We emphasise that $\Starhat$ is distinct
 from the induced Hodge dual with respect to
 the optical metric given by
 $\Starhat_{\textup{optical}} = \morph^\pstar \Star (\morph^{-1})^\pstar$. 
We recall that Maxwell's equations on $\man$ are given by \eqref{intro_Max_FG}.
Now setting
\begin{align}
\Fhat=\morph^\pstar \Ffield~.
\label{gen_Fhat}
\end{align}
Maxwell's equations on $\Mhat$ read
\begin{align}
  d\Fhat=0
\qquadand
  d ~ \Starhat ~ 
 \left( \epsilon_0 \Fhat + \Pihat \right) 
= \speedofL^{-2} {\Starhat} \Jhat~,
\label{intro_max_hat_noJ}
\end{align}
where the first of these follows immediately from  \eqref{gen_Fhat}, on account of the exterior derivative $d$ commuting with the pullback $\morph^\pstar$.
Note that in the following we will define
 $\morph^\mstar=(\morph^{-1})^\pstar$.

The derivation of $\Psihat$ and $\Phihat$ is given for a generic
transformation in section \ref{sch_Trans_Opp} below and the specific
example of the one pole Lorentz oscillator in section
\ref{S-transformedpole}.

The alternative method of representing the
constitutive relations is in terms of the integral kernel as described
in section \ref{S-integralkernel}.
The goal here is two calculate the 4--form $\kappahat$ on the product
manifold $\MXhat\times\MYhat$
~
\begin{equation}
  \kappahat 
= 
  \tfrac{1}{4} 
  \kappahat_{abcd}(x,y)
  d\xhat^a
  \wedge d\xhat^b
  \wedge d\yhat^c
  \wedge d\yhat^d
,
\label{Trans_kernal}
\end{equation}
so that
\begin{align}
\Pihat_{ab}(x) = \tfrac14\epsilon^{cdef}\int_{\MYhat}
\kappahat_{abcd}(x,y) ~  \Fhat_{ef}(y)
\,d\yhat^0\wedge d\yhat^1\wedge d\yhat^2\wedge d\yhat^3~,
\label{Trans_ker_Pi}
\end{align}
where $\MXhat$ and $\MYhat$ are two copies of $\Mhat$ corresponding to
the $\xhat$ and $\yhat$ coordinates respectively.  The calculation of
$\kappahat$ for a generic transformation is given in section
\ref{sch_Trans_Ker}, and for the one pole Lorentz oscillator in
section \ref{sch_OnePKer}.

%
\subsection{Operator representation}
\label{sch_Trans_Opp}

In this section we show how the operators $\Psi$ and $\Phi$ are
transformed under diffeomorphism (to $\Psihat$ and $\Phihat$
respectively), and demonstrate the invariance of Maxwell's equations.
We start with the operator equation for the polarisation 2--forms $\Pi$
and $\Pihat$ on $\man$ and $\Mhat$ respectively given by
\eqref{intro_constitutive_oper_unhat} and
\eqref{intro_constitutive_oper_hat}.
The only assumption we make is that $\Psi$ is {\em Killing} and {\em
  Closed}, that is that it commutes with the Hodge dual and the
exterior derivative respectively, i.e.
\begin{align}
\Psi ~ \Star = \Star ~ \Psi
\qquadand
\Psi ~ d = d ~ \Psi
.
\label{gen_Kill}
\end{align}
A sufficient condition so that \eqref{gen_Kill}
holds is that $\Psi$ is created out of lie derivatives with respect to
Killing vectors. Since $L_t=L_{\pfrac{}{t}}$ is the Lie derivative
with respect to the Killing vector $\pfrac{}{t}$ then $d\,L_t=L_t\,d$
and $\Star\, L_t=L_t\,\Star$. Hence $\Psi$ given by
\eqref{intro_constitutive_Psi_Phi_uhat} satisfies \eqref{gen_Kill}.

\begin{theorem}
Let
\begin{align}
\Psihat=
\morph^\pstar ~ \Psi ~ \Star^{-1}\morph^\mstar ~ \Starhat
,\quad
  \Phihat
=
  \epsilon_0
  \morph^\pstar \Psi\morph^\mstar 
 - \epsilon_0 ~  \Psihat
 + \morph^\pstar\Phi\morph^\mstar
\quadand
\Jhat=\Starhat^{-1}\morph^\pstar \Star \Jvec~,
\label{intro_Psihat_Phihat}
\end{align}
then
\begin{align}
  \Psi\morph^\mstar
  \left[
    d
    \Starhat
    \left( \epsilon_0 \Fhat + \Pihat \right)
   -
    \speedofL^{-2}
    \Starhat
    \Jhat
  \right]
=0~.
\label{intro_max_hat}
\end{align}
\end{theorem}
Again assuming we avoid the plasma resonances, then \eqref{intro_max_hat} is
equivalent to \eqref{intro_max_hat_noJ}.

\begin{proof}

Let $\Gfield = \epsilon_0 \Ffield + \Pi$ and
$\Ghat=\epsilon_0 \Fhat + \Pihat$.
Recall that the exterior derivative commutes with the pull-backs
$\morph^\pstar$ and $\morph^\mstar$. Thus using \eqref{gen_Fhat},
\eqref{intro_constitutive_oper_unhat}, \eqref{intro_constitutive_oper_hat}
and \eqref{gen_Kill} in turn we have
\begin{align*}
\Psi\morph^\mstar d & \Starhat ~ \left( \epsilon_0 \Fhat+\Pihat \right)
=
 \Psi d \morph^\mstar \Starhat ~ \Ghat
=
 d \Psi \morph^\mstar \Starhat ~ \Ghat
=
 d \Psi \Star \Star^{-1} \morph^\mstar \Starhat ~ \Ghat
=
 d \Star \Psi \Star^{-1} \morph^\mstar \Starhat ~ \Ghat
\\&=
 d \Star \morph^\mstar
 \left( \morph^\pstar \Star^{-1} \Psi \morph^\mstar \Starhat \right)
 \Ghat
=
 d \Star \morph^\mstar \Psihat \left( \epsilon_0 \Fhat + \Pihat \right)
=
  d \Star \morph^\mstar 
  \left( \epsilon_0 \Psihat \Fhat + \Psihat \Pihat \right)
\\&=
  d \Star \morph^\mstar 
  \left( \epsilon_0 \Psihat \Fhat + \Phihat \Fhat \right)
=
  d \Star \morph^\mstar
  \left( 
    \epsilon_0 \Psihat\Fhat
   +
    \left( 
      \epsilon_0  \morph^\pstar \Psi\morph^\mstar
     -
      \epsilon_0  \Psihat
     + 
      \morph^\pstar\Phi\morph^\mstar
    \right)
    \Fhat
  \right)
\\
&=
   d \Star \morph^\mstar
   \left( \epsilon_0  \morph^\pstar\Psi\morph^\mstar\Fhat
          +\morph^\pstar\Phi\morph^\mstar\Fhat
   \right)
=
  d \Star 
  \left(\epsilon_0  \Psi\morph^\mstar \Fhat+\Phi\morph^\mstar \Fhat \right)
\\
&=
  d \Star \left(\epsilon_0  \Psi \Ffield + \Phi F\right)
=
 d \Star \left(\epsilon_0 \Psi \Ffield + \Psi\Pi\right)
=
 d \Star \Psi \left(\epsilon_0 \Ffield + \Pi \right)
=
 d \Psi \Star \left(\epsilon_0 \Ffield + \Pi \right)
\\&=
 \Psi d \Star \left(\epsilon_0 \Ffield + \Pi \right)
=
  \speedofL^{-2}
   \Psi \Star \Jvec
=
  \speedofL^{-2}
   \Psi\morph^\mstar \morph^\pstar \Star \Jvec
=
  \speedofL^{-2}
   \Psi\morph^\mstar \Starhat\, \Jhat
.
\end{align*}
\end{proof}

If we have the inverse of $\Psi$ (for non-resonant modes) then, by
comparison with \eqref{intro_chi}, we can write down an expression
for induced susceptibility
$\chihat(\omegahat,\khat,\tHat,\xhat)=\Psihat^{-1}\Phihat$, 
 where the polarisation $\Pihat=\chihat\Fhat$. 
Again preserving the order of
 $(\omegahat,\khat,\tHat,\xhat)$  we then have
\begin{align}
  \chihat 
=
  \epsilon_0
  \left( \Zhat^{-1} - 1 \right)
 +
  \Psihat^{-1}
  \morph^\pstar
  \Phi
  \morph^\mstar
\quadtext{where}
  \Zhat
=
  \morph^\pstar
  \Star^{-1}
  \morph^\mstar
  \Starhat
\quadtext{so that}
  \Psihat
=
  \morph^\pstar
  \Psi 
  \morph^\mstar
  \Zhat~.
\label{gen_chihat}
\end{align}
However, as stated previously,
 the operator $\Psihat^{-1}$ is not a simple matrix
inverse of $\Psihat$ since the entries in $\Psihat$ depend on
$(\omegahat,\khat,\tHat,\xhat)$ which do not commute. For this article
we calculate the operators $\Psihat$ and $\Phihat$, which are the most
useful, and do not consider the operator $\Psihat^{-1}$.

Given $\morph$, one can
calculate the map
$\Zhat$ which takes 2--forms on $\Mhat$ to 2--forms on $\Mhat$ and may be
written with respect to components $\Zhat^{ab}_{cd}$ so
that for any 2--form $\alpha$, 
\begin{align}
  \Zhat \alpha
=
  \tfrac14 \Zhat^{ab}_{cd}
  ~
  d\xhat^{c} \wedge d\xhat^{d} \wedge
  \left( i_{\xhat^a} ~ i_{\xhat^b} ~ \alpha \right)
=
  \morph^\pstar
  \Star^{-1}
  \morph^\mstar
  \Starhat
  \alpha
.
\label{gen_Zhat_components}
\end{align}
The map $\Zhat$ is sufficient to encode any spatial transformation device,
 whether or not it has temporal dispersion,
 and all spacetime transformation devices without dispersion. 
The components of the constitutive tensor,
 $\Zhat^{ab}_{cd}$,
 may be written as a block $2\times 2$ tableau representing the permittivity, 
 the (inverse) permeability
 and the magneto-electric tensors (see \cite{McCall2011a}, Eq. (10)).

Using $\Phihat = \epsilon_0 \morph^\pstar \Psi \morph^\mstar
                - \epsilon_0 \Psihat
                + \morph^\pstar \Phi \morph^\mstar$ 
 and $\Psihat = \morph^\pstar ~ \Psi ~ \morph^\mstar ~ \Zhat$ 
 it is clear that the key step remaining in order to calculate $\Psihat$ and
 $\Phihat$ is to calculate the operators
 $\morph^\pstar\Psi\morph^\mstar$ and
 $\morph^\pstar\Phi\morph^\mstar$.  
This is achieved in section \ref{sch_OnePOpp}.

%
\subsection{Integral kernel representation}
\label{sch_Trans_Ker}

In this section we show how the integral kernel $\kappa$ is
transformed into $\kappahat$ under the diffeomorphism $\morph$.
The map $\morph:\Mhat\to \man$ generalises to the maps $\morphX:\MXhat\to\MX$
and $\morphY:\MXhat\to\MY$. Likewise for $\PiX$ and $\PiY$ to give the
following commutative diagram:
\begin{equation}
\begin{aligned}
\xymatrix{
&\MXhat\times\MYhat\ar[dl]_\PiXhat\ar[dr]^\PiYhat\ar@/^2pc/[rrr]^{\morphXY} & & &
\MX\ar[dl]_\PiX\times\MY\ar[dr]^\PiY &
\\
\MXhat \ar@/_2pc/[rrr]_{\morphX}&& \MYhat \ar@/_2pc/[rrr]_{\morphY}& \MX && \MY
}
\end{aligned}
\label{gen_com_diag}
\end{equation}
where $\morphXY = \morphX \times \morphY$. Let
\begin{align}
\kappahat
=
\ZhatX^{-1}\morphXY^\pstar
  \left( \epsilon_0 \Delta + \kappa \right)
 -
  \epsilon_0 \Deltahat
,
\label{gen_ker_kappahat}
\end{align}
where
\begin{align}
  \ZhatX^{-1} 
= 
  \left(
    \StarXhat^{-1} ~ \morphXY^\pstar ~ \StarX ~ \morphXY^\mstar
  \right)
,
\label{gen_ker_Z}
\end{align}
and
$\Deltahat=\frac14 \epsilon_{abcd} \delta^{(4)}(x-y)
d\xhat^{a}\wedge d\xhat^{b}\wedge d\yhat^{c}\wedge d\yhat^{d}$.
We show in lemma \ref{lm_Delta_preserved} in the Appendix that
 $\Deltahat=\morphXY^\pstar\Delta$.
Here $\StarX$ means the Hodge dual applies only to the $dx^a$ components.
Similar to \eqref{gen_Zhat_components} we can calculate the components
 of $\ZhatX^{-1}$ so that for any 4--form $\alpha$,
 which has degree 2 with respect to both $\xhat^a$ and $\yhat^a$,
 i.e. $\alpha= \tfrac14 \alpha_{abcd} ~ d\xhat^{a}\wedge d\xhat^{b}\wedge
d\yhat^{c}\wedge d\yhat^{d}$. 
We have
\begin{equation}
\begin{aligned}
(\ZhatX^{-1})\alpha
&=
\StarXhat^{-1} ~ \morphX^\pstar ~ \StarX ~ \morphX^\mstar ~ \alpha
=
\tfrac{1}{4}
  \left( \ZhatX^{-1} \right)^{ab}_{cd}
   ~ 
  d\xhat^{c} \wedge d\xhat^{d}
  \wedge
  \left( i_{\xhat^c} ~ i_{\xhat^d} ~ \alpha \right)
\\
&=
  \frac{1}{8}
  \left( \ZhatX^{-1} \right)^{ab}_{cd}\
  \alpha_{abef}
  ~
   d\xhat^{c} \wedge d\xhat^{d} \wedge d\yhat^{e} \wedge d\yhat^{f}
.
\end{aligned}
\label{gen_ZhatX_components}
\end{equation}
Here the coefficients $(\ZhatX^{-1})^{ab}_{cd}$ are the same as
$(\Zhat^{-1})^{ab}_{cd}$, i.e. the inverse of the $\Zhat^{ab}_{cd}$
given in \eqref{gen_Zhat_components}. The proof that $\kappahat$ given
by \eqref{gen_ker_kappahat} satisfies Maxwell's equation
\eqref{intro_max_hat_noJ} is given lemma
\ref{lm_kappahat_solves_Maxwell} in the Appendix.
From $\Deltahat=\morphXY^\pstar\Delta$
we can write \eqref{gen_ker_kappahat} as
\begin{align}
\kappahat
=
  \ZhatX^{-1}\morphXY^\pstar ~ \kappa
 + \epsilon_0 \ZhatX^{-1} \Deltahat 
 - \epsilon_0 \Deltahat
.
\label{gen_ker_kappahat_alt}
\end{align}
Applying \eqref{gen_ZhatX_components} we can calculate the second term
in \eqref{gen_ker_kappahat_alt} as
\begin{align}
  \ZhatX^{-1} ~ \Deltahat
=
  \frac{1}{8}
  \left (\ZhatX^{-1} \right)^{ab}_{cd} 
 ~ \delta^{(4)}(x-y)
  \epsilon_{abef}
  ~
  d\xhat^{c}
  \wedge d\xhat^{d}
  \wedge d\yhat^{e}
  \wedge d\yhat^{f}~.
\label{gen_ker_ZX_Delta_components}
\end{align}
Thus the challenge is to calculate the 4--form
$\morphXY^\pstar ~ \kappa$. This is
achieved in section \ref{sch_OnePKer}.

%
\section{Device properties}\label{S-transformedpole}

Here we consider how a one-pole resonance is transformed under
 the diffeomorphism $\morph$.
 In the operator representation it is convenient to make the
 restriction that $t$ is transformed under $\morph$ to $\tHat$. This
 restriction is relaxed in the integral kernel.
Finally, 
 in Sec. \ref{sch_OnePconstraint},
 we consider the task of implementing a transformation
 in terms of the required material properties --
 i.e. the material dynamics, and its associated dispersion.
Because of the complexity of a situation 
 involving a general transformation, 
 we utilize some simplifying assumptions 
 to clarify which parts of the transformed dynamics
 are most important.

%
\subsection{Operator representation}
\label{sch_OnePOpp}

As stated above, our goal is to calculate $\Psihat$ and $\Phihat$, so that the
polarisation 2--form $\Pihat$ satisfies \eqref{intro_constitutive_oper_hat}.
Let $\man$ have coordinates $(t,x,y,z)$ and $\Mhat$ have coordinates
$(\tHat,\xhat,\indevice{y},\indevice{z})$ with the diffeomorphism
\eqref{gen_phi} given by
\begin{equation}
\begin{aligned}
&\morph^\pstar t=\tHat
, \quad
   \morph^\pstar x=x(\tHat,\xhat)
,\quad 
  \morph^\pstar y = \indevice{y}
\quadand
   \morph^\pstar z = \indevice{z}
,
\\
\text{and inverse relations}\quad
&\morph^\mstar \tHat=t\,,\quad\morph^\mstar \xhat=\xhat(t,x)
\,,\quad \morph^\mstar \indevice{y} = y
\quadand \morph^\mstar \indevice{z} = z~.
\end{aligned}
\label{tth_diffeo_coords}
\end{equation}
Thus we have chosen to transform
 only between $(t,x)$ and $(\tHat,\xhat)$
 and furthermore we have set $t=\tHat$ so that we are
 transforming only in space,
 albeit in a time dependent manner.
Equation \eqref{tth_diffeo_coords} implies $\partial_x\,\tHat=0$,
$\partial_t\,\tHat=1$, $\partial_\xhat\, t=0$ and $\partial_\tHat\,
t=1$. Thus
\begin{equation}
\begin{aligned}
&\partial_{x}
= (\partial_x\xhat)\, \partial_{\xhat}
\,,\qquad
\partial_{t}
= (\partial_{t} \xhat)\, \partial_{\xhat} +
\partial_{\tHat}~,
\\\qquadand&
\partial_{\xhat} = (\partial_{\xhat}{x})\, \partial_{x}
\,,\qquad
\partial_{\tHat} = (\partial_{\tHat}{x})\, \partial_{x} +
\partial_{t}
.
\end{aligned}
\label{tth_diffeo_partials_hat}
\end{equation}
As stated in section \ref{sch_Trans_Opp} the key step to calculate
$\Psihat$ and $\Phihat$ is $\morph^\pstar\,\Psi\,\morph^\mstar$ and
$\morph^\pstar\,\Phi\,\morph^\mstar$. These are given by
\begin{align}
\morph^\pstar\Psi\morph^\mstar=\morph^\pstar(L_t^2+\gamma L_t + \Plasf^2)\morph^\mstar
&=
\morph^\pstar L_t^2\morph^\mstar +\gamma \morph^\pstar L_t \morph^\mstar+ \Plasf^2~,
\label{OneP_Psibar}
\end{align}
where
\begin{align}
\morph^\pstar L_t\morph^\mstar
&=
L_\tHat + (\partial_t \xhat) L_\xhat
+
(\partial_\tHat\partial_t\xhat) d\tHat\wedge i_\xhat
+
(\partial_\xhat\partial_t\xhat) d\xhat\wedge i_\xhat~,
\label{tth_Lt}
\end{align}
and
\begin{equation}
\begin{aligned}
\morph^\pstar  L_t^2\morph^\mstar
=\,&
(\partial_t^2 \xhat) L_\xhat
+
L_\tHat^2
+
2(\partial_t \xhat) L_\tHat L_\xhat
+
(\partial_t \xhat)^2 L_\xhat^2
+
2(\partial_\tHat\partial_t\xhat)
d\tHat\!\wedge\! i_\xhat L_\tHat
+
2(\partial_\xhat\partial_t\xhat)
d\xhat\!\wedge\! i_\xhat L_\tHat
\\&+
2(\partial_t \xhat)(\partial_\tHat\partial_t\xhat)
d\tHat\!\wedge\! i_\xhat L_\xhat
+
2(\partial_t \xhat)(\partial_\xhat\partial_t\xhat)
d\xhat\!\wedge\! i_\xhat L_\xhat
\\&+
\left( (\partial_t\partial_\tHat\partial_t\xhat)
+(\partial_\tHat\partial_t\xhat) (\partial_\xhat\partial_t\xhat)
\right)
d\tHat\wedge i_\xhat
+
\left(
(\partial_t \partial_\xhat\partial_t\xhat)
+(\partial_\xhat\partial_t\xhat)^2
\right)
d\xhat\wedge i_\xhat~.
\end{aligned}
\label{tth_Lt_sqr}
\end{equation}
Likewise
\begin{align}
\morph^\pstar\,\Phi\,\morph^\mstar
=
\Plasf^2
\left(
 d\tHat\wedge i_\tHat
+
(\partial_t\xhat) d\tHat\wedge i_\xhat
\right)
.
\label{OneP_Phibar}
\end{align}
In \eqref{tth_Lt}-\ref{OneP_Phibar}, 
 we can interpret the action of `$d\xhat\wedge i_\tHat$'
 as taking the polarization-like parts of what it operates on 
 (i.e. $Z$)
 and applies them to the magnetization-like parts; 
 conversely `$d\tHat\wedge i_\xhat$'
 takes the magnetization parts
 and applies them to the polarization.
We have also taken expressions such as
$\partial_\tHat \partial_t \xhat$, which means first differentiate
$\xhat(t,x)$ with respect to $t$, then consider the resulting
expression as a function of $(\tHat,\xhat)$ and differentiate those
with respect to $\tHat$. Expanding this out, we
obtain\footnote{Note that the partial derivatives do not commute,
  $[\partial_\tHat,\partial_t]\ne0$. This is because $t$ and $\tHat$
  belong to different coordinate systems.}
\begin{align}
\partial_\tHat \partial_t \xhat
=
\ppfrac{\xhat}{t}+\pfrac{x}{\tHat} \,\pqfrac{\xhat}{x}{t}~,
\label{tth_partial_tHat_t_xhat}
\end{align}
In \eqref{tth_Lt_sqr} we have third derivatives. An example expanded
out becomes
\begin{equation}
\begin{aligned}
\partial_t \partial_\tHat\partial_t\xhat
&=
\frac{\partial^3 \xhat}{\partial t^3}
+
\ppfrac{x}{\tHat}\,\pqfrac{\xhat}{x}{t}
+
\pfrac{\xhat}{t}\,\pqfrac{x}{\xhat}{\tHat}\,\pqfrac{\xhat}{x}{t}
+
\pfrac{x}{\tHat}\,\frac{\partial^3 \xhat}{\partial x\partial t^2}~,
\end{aligned}
\label{tth_partials_t_tHat_t_xhat}
\end{equation}
The reason for the third derivative of the coordinate transformation is
that since $F$ is a 2--form, we need one derivative of the coordinate
transformation which needs to be differentiated twice more since $\Psi$
contains $L_t^2$. The proof of \eqref{OneP_Psibar}-\eqref{OneP_Phibar}
is given in the Appendix, where we no longer impose $t=\tHat$. The general
transformation $(t,x)$ to $(\tHat,\xhat)$ includes, for example, the curtain map introduced in the original proposal for the spacetime cloak \cite{McCall2011a}, which utilized a Lorentz boost.

%
\subsection{Integral kernel representation}
\label{sch_OnePKer}

In the integral kernel representation, our goal is to calculate
$\kappahat$, so that the polarisation 2--form $\Pihat$ given by
\eqref{Trans_ker_Pi} satisfies Maxwell's equation
\eqref{intro_max_hat_noJ}.

\begin{figure}
\centering{
\setlength{\unitlength}{0.04\textwidth}
\begin{picture}(18,11.5)(-5,0)
\put(-5,1){\begin{picture}(4,4)
   \put(0,1){\vector(0,1){4}}
   \put(0,1){\vector(4,1){4}}
   \put(0,1){\vector(4,-1){4}}
   \put(-0.9,3){$\yhat^0$}
   \put(4,1.5){$\yhat^1$}
   \put(4,0){$\yhat^2,\yhat^3$}
\end{picture}}
\put(0,0){
   \includegraphics[width=10\unitlength]{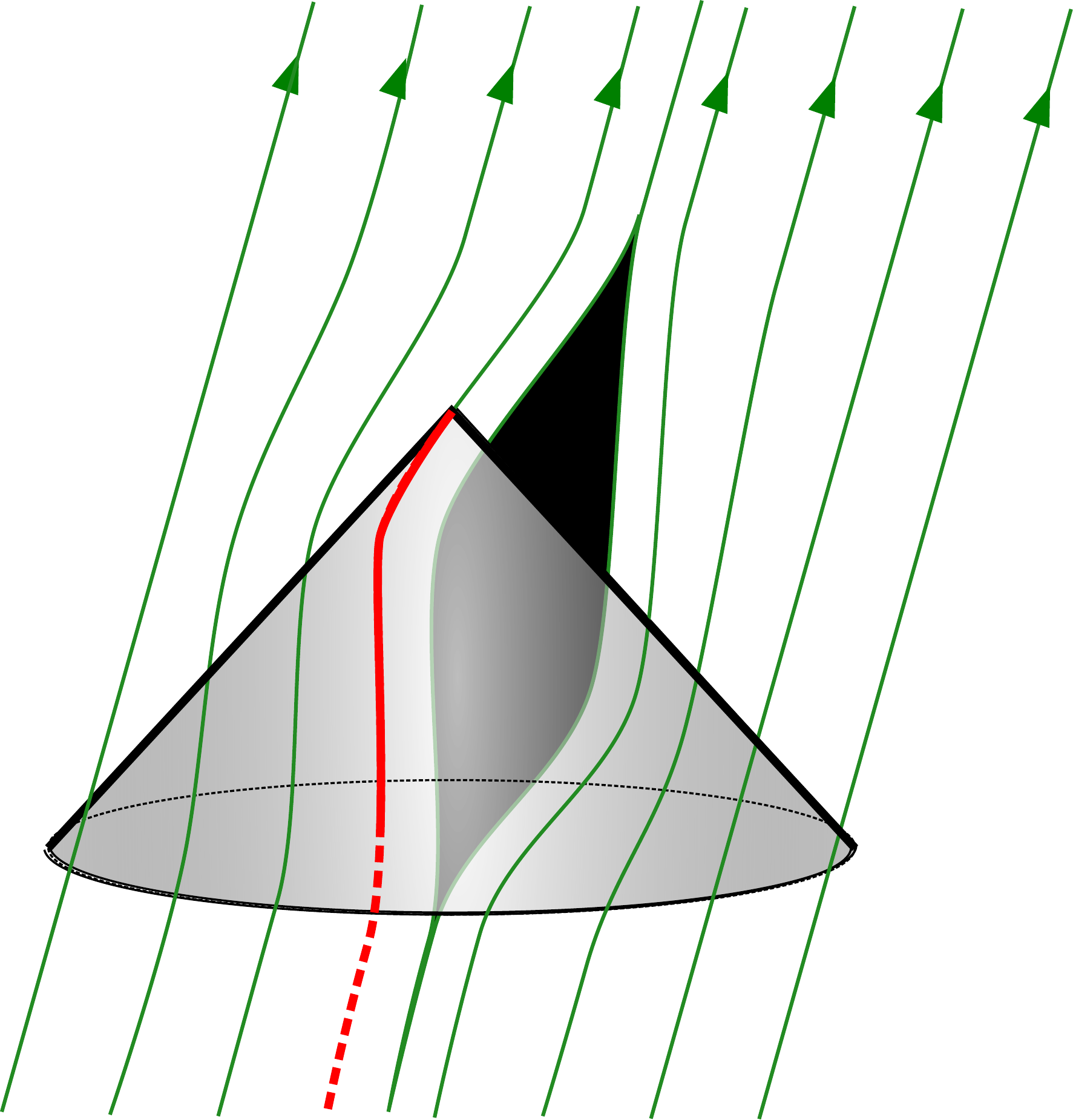}}
\put(3.4,6.8){$(\xhat^0,\underline{\xhat})$}
\end{picture}
}
\caption{The support (red) in $\MYhat$ of the kernel
  $\morphXY^\pstar\kappa$ and $\kappahat$
  \eqref{OnePKer_phistar_kappa}, for the point
  $(\xhat^0,\underline{\xhat})\in\MX$. Here the green lines correspond
  to the lines $\Set{\phi(y^0,\Vy)|y^0\in\Real}$ for fixed $\Vy$.}
\label{fig_kernel_support_real}
\end{figure}

The transformation of the kernel is much easier. We simply set
\begin{align}
\morphXY^\pstar\kappa
=
\chi_0\,\frac{\theta \left( x^0-y^0 \right) \delta^{(3)}(\Vx-\Vy)}{2(\sigma_+-\sigma_-)}
 \left(  e^{\sigma_+(x^0-y^0)} - e^{\sigma_-(x^0-y^0)}  \right) 
\
\epsilon_{\lambda\nu\rho}\
dx^{0} \wedge dx^{\lambda} \wedge dy^{\nu}~ \wedge dy^{\rho}~,
\label{OnePKer_phistar_kappa}
\end{align}
where we regard
\begin{align}
x_a=x_a(\xhat)
\,,\quad
y_a=y_a(\yhat)
\,,\quad
dx^a=\pfrac{x^a}{\xhat^b} d\xhat^b
\quadand
dy^a=\pfrac{y^a}{\yhat^b} d\yhat^b~.
\label{OnePKer_phistar_kappa_x_dx}
\end{align}
By substituting \eqref{OnePKer_phistar_kappa_x_dx} into
\eqref{OnePKer_phistar_kappa} we can generate
$\morphXY^\pstar\,\kappa$ and hence using \eqref{gen_ker_kappahat_alt}
we can generate $\kappahat$.  It is clear that having done this,
the components of $\kappahat$ will not be functions of
$\xhat-\yhat$. This is as expected since
$\kappahat$ is not the Fourier transform
of a function with arguments $(\omegahat,\khat)$. 
The support of $\kappahat$ in \eqref{OnePKer_phistar_kappa}
 is given in fig. \ref{fig_kernel_support_real}.

In this representation it is easy to see that
 we are effectively doing transformation optics
 intended to mimic a dispersive medium.
As described in section \ref{sch_Trans_Ker},
 we have restricted the design medium to be homogeneous and isotropic,
 with the causal and light cone structure of Minkowski spacetime.  
The integral kernel representation
 for more general dispersive design media entails complications
 that are beyond the scope of this initial work,
 and will form the basis of further study.

%
\subsection{Practical Constraints}
\label{sch_OnePconstraint}

We expect it to be extremely challenging to construct the
 spatio-temporal dispersive medium required for
 perfect event cloaking,
 and indeed this will likely even be true of simpler
 spacetime transformation devices \cite{McCall2013a}.
However,
 in practise we will likely require only that its performance
 is better than some chosen benchmark.
The question then to be answered here is a more practical one --
 if we only have ordinary temporally dispersive materials available,
 and we wish to keep the cloaking within some performance range,
 what constraints does that put on the
 device design and operating parameters?
Here we will consider how the transformation gradients
 induced by a design morphism,
 as indicated by basic scales and sizes of the device,
 affect the magnitude of the additional terms
 that appear for the transformed medium, 
 as calculated in \ref{sch_OnePOpp}.

We assume a morphism involving a time-dependent spatial distortion,
 but with matched time coordinates (i.e. $t=\indevice{t}$).
We assume a plane-polarized light,
 the light (of course) being purely transverse in nature;
 it is a quasi continuous wave,
 centered about a frequency $\omega$,
 and with wavelength in the design medium of $\lambda$
 and wavevector $k=\lambda/2\pi$.
Specifically, 
 we assume the electric field 
 has the carrier-envelope form 
 given in \eqref{eqn-fieldenvelope}; 
 and
 although it is now common to allow for a 
 wideband field envelope \cite{Genty-KKD-2008oe},
 here we make the more stringent SVEA discussed above.

Before assuming that the transformation is gradual
 (i.e. making the GTA), 
 let us first consider the likely scope and deformation
 of spacetime points caused by a cloaking morphism, 
 taking the experiment of Fridman et al. \cite{Fridman2012}
 as a guide.
For a 1D waveguide based event cloak,
 there are only three important parameters --
 the distance $L$ over which the cloak opens and closes,
 the maximum shift $\Delta$ in position of any point,
 and the average speed of light $\speedofL_o$ within the cloak.
Using \cite{Fridman2012}, 
 we therefore set $L=200$m
 and $\Delta=5$mm.
These distances are proportional to the time scales over which
 the cloak operates:
 it has a total duration of $T = L / \speedofL_o \sim 1 \mu$s,
 and opens a `time gap' of $\tau = \Delta / \speedofL_o \sim 25$ps.

Subject to these assumptions, 
 we can write out
 the morphed/device polarization response/operator
 using \eqref{tth_Lt} and  \eqref{tth_Lt_sqr}
 in \eqref{intro_constitutive_uhat}.
Before then, 
 however, 
 we will consider the prefactors of each term, 
 each of which relates to derivatives of the morphism,
 and estimate their sizes.
This will make it clear which terms we might need to keep,
 and which to discard:
 obviuosly we expect that the first order terms will be most significant, 
 we also calculate the second order ones for completeness.
These prefactors,
 with $\indevice{t} \equiv \indesign{t}$,
 are
~
\begin{align}
  a_1 &= \partial_{t} {\indevice{x}}
&\sim
  \Delta / T
&\sim   5 \times 10^{-3} / 10^{-6}
 \qquad 
&\simeq&~~
 ~5\times 10^{3}
 ~~\textrm{[m/s]}~,
 \label{eqn-constraint-a1number}
\\
  b_1 &= \partial_{t}^2 {\indevice{x}}
&\sim
  \Delta / T^2
&\sim   5 \times 10^{-3} / 10^{-12}
 \qquad 
&\simeq&~~
  ~5 \times 10^{9}~,
 ~~\textrm{[m/s$^2$]}
 \label{eqn-constraint-b1number}
\\
  b_2 &= \partial_{x} \partial_{t} {\indevice{x}}
&\sim
  \Delta / L T
&\sim   5 \times 10^{-3} / 200\times 10^{-6}
&\simeq&~~
  25
 ~~\textrm{[1/s]}~
 \label{eqn-constraint-b2number}
.
\end{align}
We now need to combine these with the properties
 of the design medium and the electromagnetic field, 
 and compare the results to the benchmark frequency scale.
Consequently, 
 let us now assume that our cloak
 is also trying to mimic a medium with resonant frequency
 of about 300THz
  (i.e. $\Res \sim 2\times 10^{15} \textrm{rad/s}$),
 and a decay rate of nanoseconds
  ($\Damp \sim 10^{9} \textrm{s}^{-1}$).
We further assume that the illumination frequency is comparable
 to this resonant frequency,
 at $200THz$ (1432nm)
 (i.e. $\omega_0  \sim 1.4 \times 10^{15} \textrm{rad.s}^{-1}$),
 we therefore also know that allowing for
 the refactive index of silica, 
 that
 $k_0 = \omega_0 / \speedofL_{\textrm{silica}} \sim 0.7 \times 10^{7}$.
Thus
~
\begin{align}
  k a_1 / \omega
&\sim
  4 \times 10^{10}
  \textrm{~rad/s}
  / 1.4 \times 10^{15} \textrm{~rad/s}
~&\simeq&~~
  ~3 \times 10^{-5}
,
\label{eqn-constraint-ka1}
\\
  2 \pi
  b_1 / \speedofL \omega
&\sim
  2 \pi \times
  5 \times 10^{9}
  \textrm{~m/s}^2
  / 4 \times 10^{23} \textrm{~m.rad}\textrm{/s}^2
~&\simeq&~~
 ~7 \times 
 10^{-14}
,
\label{eqn-constraint-kb1}
\\
  b_2 / \omega
&\sim
  25
  \textrm{~rad}\textrm{/s}
  / 1.4 \times 10^{15} \textrm{~rad/s}
~&\simeq&~~
  ~2 \times 10^{-14}
\label{eqn-constraint-kb2}
.
\end{align}

Clearly, 
 by far
 the most significant correction here is the first order $a_1$ term, 
 which appears in \eqref{tth_Lt} and  \eqref{tth_Lt_sqr}
 in concert with the spatial derivative $L_{\indevice{x}}$; 
 keeping only these is consistent with the GTA.
The very small size of the second order terms also assures us 
 that ignoring the third order terms was reasonable. 
This means that the $a_1$-dependent
 additional dynamics of the polarization field
 is dependent on spatial gradients,
 that is,
 that each (originally independent)
 Lorentz oscillator now is affected by its neighbourhood.
However, 
 since this dominant correction is of order $10^{-5}$, 
 we can see that the Fridman et al. experiment \cite{Fridman2012}
 was safely in the regime where the spacetime effects discussed
 in this paper are negligible.
However, 
 if we hoped to extend the cloaking interval to beyond $25$ns instead of ps
 ($\Delta=5$m)
 without also extending the total cloaking length $L$,
 then we could be pushing the limitations of the GTA, 
 and more care would be required.

Nevertheless, 
 retaining the $a_1$ parts as being the primary (GTA) corrections, 
 we can now write a corrected (transformed)
 differential operator for the polarization dynamics,
 which is
~
\begin{align}
  \indevice{\Psi}_{\textrm{GTA}}
&=
  \left(
    L_{\indevice{t}}
   +
    a_1
    L_{\indevice{x}}
  \right)^2
 +
  \gamma
  \left(
    L_{\indevice{t}}
   +
    a_1
    L_{\indevice{x}}
  \right)
 +
  \Res^2
.
\end{align}
Under this level of approximation, 
 we can see that although the device medium has spatial dispersion 
 in addition to the required temporal dispersion, 
 there is no dynamic magnetoelectric cross-coupling between 
 the polarization part of $\indevice{Z}$
 and its magnetization part (and vice versa).
However, 
 the transformed driving term $\indevice{\Phi}$
 in \eqref{OneP_Phibar}
 does have a magnetic driving contribution of order $a_1$; 
 the polarization and magnetization thus evolve independently, 
 but with matching dynamics.

Although an analysis of the differential equation governing
 the material response is perhaps the most physical approach
 \cite{Kinsler-2011ejp},
 it is rather more common to look at the material response
 in the frequency/wavevector domain,
 and consider the dispersion relations.
To get these
 we use the SVEA, 
 where the nearly CW nature of the illuminating field
 means that 
 its effect on the polarization is dominated by the
 exponential carrier oscillations,
 i.e. by $\omega_0$ and $k_0$.
It can now be easily seen that
\begin{align}
  \morph^\pstar\Psihat\morph^\mstar 
&=
 -
  \left(\omegahat-a_1\khat\right)^2 
 +
  i
  \Damp
  \left(\omegahat-a_1\khat\right)
 +
  \Res^2\,,
\label{Appr_morphPhi}
\end{align}
so that
\begin{align}
  \chihat 
&=
  \epsilon_0 \, 
  \left( \Zhat^{-1} - 1\right)
 +
  \frac{\Plasf^2}
       {-\left(\omegahat-a_1\khat\right)^2 
        +i\Damp\left(\omegahat-a_1\khat\right)
        +\Res^2
       }
  ~.~
  \Zhat^{-1}
  \left(
    d\tHat\wedge i_\tHat
   +
    a_1 d\tHat\wedge i_\xhat
  \right)
.
\label{Appr_chihat}
\end{align}
The components of the inverse map $\Zhat^{-1}$
 can easily be calculated. 
For example,
 from the first term in \eqref{Appr_chihat},
 we see the electric field in the $\yhat$-direction
 gives rise to a magnetisation in the $\zhat$-direction
 as well as a polarisation in the $\yhat$-direction since
\begin{align}
  \Zhat^{-1}(d\tHat\wedge d\yhat)
=
  (\partial_\tHat x)(d\xhat\wedge d\yhat)
 +
  (\partial_\xhat x)(d\tHat\wedge d\yhat)
\label{Appr_Zhat}
,
\end{align}
 where we should note that the derivatives here are of $x$, 
 and not of $\xhat$ as previously
 (i.e. as for $a_1$, $b_1$, $b_2$).
Thus
\begin{align}
  \epsilon_0 \,d\tHat\wedge d\yhat + \chihat(d\tHat\wedge d\yhat)
&=
  \left[
    \epsilon_0
   +
    \frac{\Plasf^2}
         {-\left(\omegahat-a_1\khat\right)^2
          +i\Damp\left(\omegahat-a_1\khat\right)
          +\Res^2
          }
  \right]
\nonumber
\\
&\qquad\qquad\qquad\qquad \times
  \left[
     (\partial_\tHat x)(d\xhat\wedge d\yhat)
    +
     (\partial_\xhat x)(d\tHat\wedge d\yhat)
  \right]
\label{Appr_chihat_elec}
\end{align}
We emphasize again,
 that despite the application of both the GTA and the SVEA,
 the resultant medium is still spatially and temporally dispersive
 and magneto-electric.

As we would expect, 
 this is more complicated than that of the simple one pole resonance 
 specified by the design requirement in \eqref{intro_one_pole_resonance}.
We can see that
 the morphism induces
 an effective frequency shift $\propto a_1 k_0$,
 and the loss term
 is augmented by an induced evanescence $\propto \gamma a_1$;
 we now have spatial effects which can be considered
 as spatial dispersion.
For a symmetric cloak,
 we would expect that the $a_i$ coefficient
 would change sign at the half-way point,
 for zero net effect.
In a perfect cloak we will expect the effect of the modified medium 
 to cancel out --
 indeed that is the point of the transformation design process --
 but in any imperfect implementation
 that may not be the case.

Although it is easy enough to calculate the sizes of these 
 morphism-induced changes to the medium, 
 and discuss the significance of their magnitudes, 
 it is less easy to infer their integrated effect 
 as some illuminating field passes through the device medium.
In fig. \ref{fig-diagram} we indicated
 that a spacetime cloak built without the additional material dynamics
 would hide the event,
 and not return any evidence of what it was,
 but that --
 for pictorial purposes, 
 at least --
 a diverted illuminating pulse would
 end up with the wrong chirp.
This mismatch could raise suspicions
 that all was not as it should seem.

However,
 what the signature would be in practise is less clear,
 as the effect on the illuminating field at any given position
 depends on the current local state of that medium.
Further, 
 that state is {also} an integral over its past, 
 a past that involves its internal dynamics, 
 the propagating and changing field,
 \emph{and} an evolving morphism.
Since a detailed simulation of the polarization and field dynamics
 is beyond the scope of this paper,
 we instead have focussed on the magnitudes of the non-Lorentzian $a_1$ terms,
 and whether or not the effects of neglecting them will or will not
 be significant.
This is because
 from a practical perspective,
 they mean we have the challenging task of constructing a metamaterial
 that naturally follows such a non-Lorentzian dynamics,
 but we would prefer to keep our device simpler and more achievable.

%
\section{Conclusion}\label{S-conclusion}

In this paper we have
 for the first time
 accounted for dispersion in spacetime cloaking. 
Since dispersion is common to all media
 through which electromagnetic waves propagate,
 this advance must be regarded as significant. 
In particular,
 we have given
 the design recipe for fooling an observer
 into thinking that the electromagnetic signal
 she receives has travelled through a uniform one-pole resonance medium,
 whereas in reality,
 the electromagnetic signal has been distorted and reformed
 so as to permit events to occur undetected for a brief period.

As we have seen,
 the problem presented a considerable mathematical challenge,
 one that we have addressed using both 
 an operator and an integral kernel method, 
 applied within a framework which treats
 the electromagnetic field as a differential form. 
Some simplifying assumptions enabled us to make further progress. 
We identified the GTA and SVEA as being approximations that allow 
 us to ignore several of the terms generated under spacetime transformation. 
Nevertheless,
 even when both these approximations are operative, 
 the inherent nature of the transformation is such that space and time
 become mixed,
 so that the medium prescribed by the transformation
 is one that must have
 defined inhomogeneous temporal and inhomogeneous spatial dispersion.
The metamaterial design will still be a considerable technological challenge. 
Rather than attempt to address how these challenges can be met, 
 we instead focussed on using our results to estimate 
 the likely impact of imperfections on practical spacetime cloaks. 
We did this by estimating the size of various terms using 
 the experimental parameters of 
 Fridman et al's spacetime cloak \cite{Fridman2012}
 as a guide. 
We found that the dominant dynamics of the induced polarization field 
 depend on spatial gradients,
 and that at this level of approximation magneto-electric effects 
 can be ignored.

Our rather complete analysis of the problem of achieving spacetime cloaking
 using dispersive media will likely assist experimentalists
 in their design of spacetime cloaks. 
A key step achieved here is that in each case the limitations 
 of any practical design can be quantitatively assessed.

\section*{Acknowledgements}
PK and JG are grateful for the support provided by STFC (the
Cockcroft Institute ST/G008248/1) and EPSRC (the Alpha-X project
EP/J018171/1).
RTT was grateful for the support from the
 Royal Society of New Zealand
 through the Marsden Fund Fast Start Grant no. UOO1219.
MWM is grateful for support from EPSRC
 grant number EP/K00305/1.

\section*{References}


\providecommand{\newblock}{}

%
\appendix

\section{Additional Lemmas}

In this appendix we demonstrate the results used in this paper. We
use some additional notation. We use multi-indices on forms so that
$dx^{ab}=dx^a\wedge dx^b$.

\begin{lemma}
\label{lm_ker_soln}
The kernel given by \eqref{intro_ker_kappa} satisfies the operator
equation \eqref{intro_ker_green}.
\end{lemma}
\begin{proof}
From \eqref{intro_constitutive_Psi_Phi_uhat}
the operators
\begin{align*}
\PsiX=L_{x^0}^2 + \gamma L_{x^0} + \Res^2
\qquadand
\PhiX=\Plasf^2\, d{x^0}\wedge i_{x^0}
\end{align*}
Acting $\PsiX$ on $\kappa$ gives
\begin{align*}
&\PsiX\,\kappa/\Plasf^2
\nonumber
\\
&=
\left(L_{x^0}^2 + \gamma L_{x^0} + \Res^2\right)
\nonumber
\\
&\qquad \times 
  \left(
    \frac{\theta(x^0-y^0)\,\delta^{(3)}(\Vx-\Vy)}
         {2(\sigma_+-\sigma_-)}
    \left( e^{\sigma_+(x^0-y^0)} - e^{\sigma_-(x^0-y^0)} \right)
   \epsilon_{\mu\nu\rho}\
   dx^{0\mu}\wedge dy^{\nu\rho}
  \right)
\\&=
  \frac{\delta^{(3)}(\Vx-\Vy)}
       {2(\sigma_+-\sigma_-)}
  \,
  \left( L_{x^0}^2 + \gamma L_{x^0} + \Res^2 \right) 
  \left(
    \theta(x^0-y^0)
     \left(  e^{\sigma_+(x^0-y^0)} - e^{\sigma_-(x^0-y^0)}  \right) 
  \right)
  \epsilon_{\mu\nu\rho}\
  dx^{0\mu}\wedge dy^{\nu\rho}
\\&=
  \frac{\delta^{(3)}(\Vx-\Vy)}
       {2(\sigma_+-\sigma_-)}\,
\left[
  \theta(x^0-y^0)
  \left( L_{x^0}^2 + \gamma L_{x^0} + \Res^2 \right)
  \left(
    e^{\sigma_+(x^0-y^0)} - e^{\sigma_-(x^0-y^0)} 
  \right)
\right.
\\&\qquad\qquad\qquad\qquad
+
  \left(
    \left( L_{x^0}^2 + \gamma L_{x^0} \right)
    \theta(x^0-y^0)
  \right)
    \left(
      e^{\sigma_+(x^0-y^0)} - e^{\sigma_-(x^0-y^0)} 
   \right)
\\&\qquad\qquad\qquad\qquad
\left.
+
2 \left( L_{x^0}\theta(x^0-y^0) \right) 
  \left( L_{x^0}
    \left( e^{\sigma_+(x^0-y^0)} - e^{\sigma_-(x^0-y^0)} \right)
  \right)
\right]
\epsilon_{\mu\nu\rho}\
dx^{0\mu}\wedge dy^{\nu\rho}
\\&=
\frac{\delta^{(3)}(\Vx-\Vy)}{2(\sigma_+-\sigma_-)}\,
\left[
  \left(
    \delta'(x^0-y^0)+\gamma\delta(x^0-y^0)
  \right)
  \left(
    e^{\sigma_+(x^0-y^0)} - e^{\sigma_-(x^0-y^0)} 
  \right)
\right.
\\&\qquad\qquad\qquad\qquad
\left.
 +
  2
  \left( \delta(x^0-y^0) \right)
  \left( L_{x^0}
    \left( e^{\sigma_+(x^0-y^0)} - e^{\sigma_-(x^0-y^0)} \right)
  \right)
\right]
\epsilon_{\mu\nu\rho}\
dx^{0\mu}\wedge dy^{\nu\rho}
\\&=
  \frac{\delta^{(3)}(\Vx-\Vy)}
       {2(\sigma_+-\sigma_-)}\,
  \left[
    \left( \delta(x^0-y^0) \right)
    \left( L_{x^0}
      \left(
        e^{\sigma_+(x^0-y^0)} - e^{\sigma_-(x^0-y^0)}
      \right)
    \right)
  \right]
  \epsilon_{\mu\nu\rho}\
  dx^{0\mu} \wedge dy^{\nu\rho}
\\&=
  \frac{\delta^{(3)}(\Vx-\Vy)}
       {2} \,
  \delta(x^0-y^0)
  \epsilon_{\mu\nu\rho}\
  dx^{0\mu}  \wedge dy^{\nu\rho}
=
  \tfrac12
  \epsilon_{\mu\nu\rho}
  \delta^{(4)}(x-y)
  dx^{0\mu} \wedge dy^{\nu\rho}
\end{align*}
where we have used the identity that $\delta(z)f(z)=0$ and
$\delta'(z)f(z)=-f'(0)$ for any functions $f(z)$ such that $f(0)=0$.
Acting $\PhiX$ on $\Delta$ gives
\begin{align*}
\PhiX\,\Delta
&=
\tfrac14 \epsilon_{abcd} \PhiX\,
\delta^{(4)}(x-y) dx^{ab}\wedge dy^{cd}
=
\tfrac14 \epsilon_{abcd} \Plasf^2\, \delta^{(4)}(x-y)
d{x^0}\wedge i_{x^0}
(dx^{ab}\wedge dy^{cd})
\\&=
\tfrac14 \epsilon_{abcd} \Plasf^2\, \delta^{(4)}(x-y)
d{x^0}\wedge (\delta_0^a dx^b - \delta_0^b dx^a)\wedge dy^{cd}
\\&=
\tfrac12 \epsilon_{0bcd} \Plasf^2\, \delta^{(4)}(x-y)
d{x^{0b}}\wedge dy^{cd}
=
\tfrac12 \epsilon_{\mu\nu\rho} \Plasf^2\, \delta^{(4)}(x-y)
d{x^{0\mu}}\wedge dy^{\nu\rho}
\end{align*}
Hence \eqref{intro_ker_green}.

\end{proof}


\begin{lemma}
\label{lm_Delta_preserved}
The 4--form distribution $\Delta$ is preserved under the map
$\morphXY:\MXhat\times\MYhat\to\MX\times\MY$, that is
\begin{align}
\morphXY^\pstar\Delta=\Deltahat
\label{gen_ker_phistar_Delta}
\end{align}
\end{lemma}
\begin{proof}
We show \eqref{gen_ker_phistar_Delta} is true by acting on an
arbitrary test form $\alpha$ in $\MX\times\MY$ to form the integral
$\int_{\MX\times\MY} \alpha\wedge\Delta$.
Note that if $\alpha$ is not of the degree (2,2), i.e.
$\alpha\ne\tfrac14\alpha_{abcd}dx^{ab}\wedge dy^{cd}$ then
$\int_{\MX\times\MY} \alpha\wedge\Delta=0$.

Thus we set $\alpha=\tfrac14\alpha_{abcd}dx^{ab}\wedge dy^{cd}$.
Since $\alpha$ is a test form it has
 compact support,
 so that $\int_{\MX\times\MY} \alpha\wedge\Delta$ is finite.
\begin{align*}
\int_{\MX\times\MY} \alpha\wedge\Delta
&=
\frac1{16} \int_{\MX\times\MY}\alpha_{abcd}(x,y)\,\delta^{(4)}(x-y)\,
dx^{ab}\wedge dy^{cd}\wedge dx^{ef}\wedge dy^{gh}\epsilon_{efgh}
\\&=
\frac1{16} \int_{\MX\times\MY}\alpha_{abcd}(x,y)\,\delta^{(4)}(x-y)\,
dx^{abef}\wedge dy^{0123}\,\epsilon_{efgh}\,\epsilon^{cdgh}
\\&=
\frac1{16} \int_{\MX}\alpha_{abcd}(x,x)
dx^{abef}\,\epsilon_{efgh}\,\epsilon^{cdgh}
\\
&=
\frac1{8} \int_{\MX}\alpha_{abcd}(x,x)
dx^{abef}(\delta_e^c\delta_f^d-\delta_f^c\delta_e^d)
\\&=
\frac1{4} \int_{\MX}\alpha_{abcd}(x,x)
dx^{abcd}=\int_{\MX} \gamma_{\textup{X}}^\pstar\alpha
\end{align*}
where $\gamma_{\textup{X}}:\MX\to(\MX\times\MY)$ with
$\gamma_{\textup{X}}(x)=(x,x)$.
Likewise setting
$\indevice\gamma_{\textup{X}}:\MXhat\to(\MXhat\times\MYhat)$,
$\indevice\gamma_{\textup{X}}(\xhat)=(\xhat,\xhat)$ then
$\gamma_{\textup{X}}\circ\morphX=\morphXY\circ\indevice\gamma_{\textup{X}}$.
Thus
\begin{align*}
\int_{\MXhat\times\MYhat} (\morphXY^\pstar \,\alpha)\wedge\Deltahat
&=
\int_{\MXhat} \indevice\gamma_{\textup{X}}^\pstar \,\morphXY^\pstar\, \alpha
=
\int_{\MXhat} \morphX^\pstar\, \gamma_{\textup{X}}^\pstar \alpha
=
\int_{\MX} \gamma_{\textup{X}}^\pstar \alpha
=
\int_{\MX\times\MY} \alpha\wedge\Delta
\\&=
\int_{\MXhat\times\MYhat} \morphXY^\pstar (\alpha\wedge\Delta)
=
\int_{\MXhat\times\MYhat} (\morphXY^\pstar\, \alpha)\wedge(\morphXY^\pstar\, \Delta)
\end{align*}
Since this is true for all $\alpha$ we have
\eqref{gen_ker_phistar_Delta}.
\end{proof}


\begin{lemma}
\label{lm_kappahat_solves_Maxwell}
The transformed kernel $\kappahat$ given
by \eqref{gen_ker_kappahat} satisfies Maxwell's equation
\eqref{intro_max_hat_noJ}.
\end{lemma}

\begin{proof}
There is a slight subtlety with regard to the manifold that $\Ffield$ lies
in. 
Let $\FX=\tfrac12 \Ffield_{ab}(x) dx^{ab}$
 and $\FY=\tfrac12 \Ffield_{ab}(y) dy^{ab}$
 be the same electromagnetic 2--form on $\MX$ and $\MY$
 respectively. 
This gives $\FX=\int_{\MY}\Delta\wedge\PiY^\pstar\FY$. 
Likewise with $\FhatX$ and $\FhatY$.

If $\alpha$ and $\beta$ are forms on $\MX\times\MY$,
 but $\beta$ only contains $dy$,
 then from \eqref{gen_ker_Z}
\begin{align*}
\ZhatX^{-1} (\alpha\wedge\beta)
&=
\StarXhat^{-1}\,\morphXY^\pstar\,\StarX\,\morphXY^\mstar
(\alpha\wedge\beta)
=
\StarXhat^{-1}\,\morphXY^\pstar\,\StarX
(\morphXY^\mstar\alpha\wedge\morphXY^\mstar\beta)
\\&=
\StarXhat^{-1}\,\morphXY^\pstar
(\StarX\,\morphXY^\mstar\alpha\wedge\morphXY^\mstar\beta)
=
\StarXhat^{-1}
(\morphXY^\pstar\,\StarX\,\morphXY^\mstar\alpha\wedge
\morphXY^\pstar\,\morphXY^\mstar\beta)
\\&=
\StarXhat^{-1}
(\morphXY^\pstar\,\StarX\,\morphXY^\mstar\alpha\wedge
\beta)
=
\StarXhat^{-1}
\morphXY^\pstar\,\StarX\,\morphXY^\mstar\alpha\wedge
\beta
=
\ZhatX^{-1}\alpha\wedge
\beta
\end{align*}
Hence we have
\begin{align*}
  \speedofL^{-2}\,\Starhat\,\Jhat
&=
  \speedofL^{-2}\,\morphX^\pstar \,\Star\,\Jvec
=
  \morphX^\pstar ~ d ~ \StarX ~ (\epsilon_0\FX+\Pi)
=
  d ~ \morphX^\pstar
  \StarX
  \int_{\MY}
  \left( \epsilon_0 \Delta + \kappa \right) 
  \wedge \PiY^\pstar
  \FY
\\
&=
  d ~ \morphX^\pstar
  \StarX
  \morphX^\mstar
  \morphX^\pstar
  \int_{\MY} 
  \left( \epsilon_0 \Delta + \kappa \right)
  \wedge\PiY^\pstar 
  \FY
\\&=
  d ~ \morphX^\pstar
  \StarX
  \morphX^\mstar
  \int_{\MYhat} 
  \morphXY^\pstar
  \left(
    \left( \epsilon_0 \Delta + \kappa \right) 
    \wedge \PiY^\pstar 
    FY
  \right)
\\&=
d ~ \StarXhat ~ \StarXhat^{-1} \morphX^\pstar ~ \StarX ~ \morphX^\mstar\
\int_{\MYhat} \morphXY^\pstar ~  \left( (\epsilon_0\,\Delta+\kappa)
\wedge\PiY^\pstar \FY \right) 
\\&=
d ~ \StarXhat ~ Z^{-1}\,
\int_{\MYhat} \morphXY^\pstar\
 \left( (\epsilon_0\,\Delta+\kappa)\wedge\PiY^\pstar \FY \right) 
\\&=
d ~ \StarXhat\
\int_{\MYhat} \ZhatX^{-1} \morphXY^\pstar\
 \left( (\epsilon_0\,\Delta+\kappa)\wedge\PiY^\pstar \FY \right) 
\\&=
d ~ \StarXhat\
\int_{\MYhat}  
  \left( \ZhatX^{-1} \morphXY^\pstar 
    \left( \epsilon_0 \Delta + \kappa \right) 
  \right)
  \wedge
  \left( \morphXY^\pstar ~ \PiY^\pstar \FY \right) 
\\&=
d ~ \StarXhat\
\int_{\MYhat} (\epsilon_0\Deltahat+\kappahat)\wedge\
 \left( \PiYhat^\pstar  ~ \morphY^\pstar ~ \FY \right) 
\\&=
  d ~ \StarXhat\
  \int_{\MYhat} 
    \left(\epsilon_0\,\Deltahat+\kappahat\right)\wedge\
    \left(\PiYhat^\pstar  ~ \FhatY \right) 
=
d ~ \StarXhat ~ \left( \epsilon_0 \FhatX + \Pihat \right)
\end{align*}
\end{proof}


\section{Derivation of the transformation $(t,x)\to(\tHat,\xhat)$}

Here we demonstrate the formulae \eqref{tth_Lt} and
\eqref{tth_Lt_sqr}. Generalise the transformation so that
\begin{equation}
\begin{aligned}
&\morph^\pstar t=t(\tHat,\xhat)\,,\quad \morph^\pstar x=x(\tHat,\xhat)
\,,\quad \morph^\pstar y = \indevice{y}
\quadand \morph^\pstar z = \indevice{z}
\\
\text{and inverse relations}\qquad
&\morph^\mstar \tHat=\tHat(t,x)\,,\quad\morph^\mstar \xhat=\xhat(t,x)
\,,\quad \morph^\mstar \indevice{y} = y
\quadand \morph^\mstar \indevice{z} = z
\end{aligned}
\label{intro_diffeo_coords}
\end{equation}
This implies
\begin{align}
\morph_\pstar^{-1} \partial_{x}
= (\partial_x\xhat)\, \partial_{\xhat} +
(\partial_x\tHat)\, \partial_{\tHat}
\qquadand
\morph_\pstar^{-1} \partial_{t}
= (\partial_{t} \xhat)\, \partial_{\xhat} +
(\partial_{t} {\tHat})\, \partial_{\tHat}
\label{intro_diffeo_partials_unhat}
\end{align}
and
\begin{align}
\morph_\pstar\partial_{\xhat} = (\partial_{\xhat}{x})\, \partial_{x} +
(\partial_{\xhat}{t})\, \partial_{t}
\qquadand
\morph_\pstar\partial_{\tHat} = (\partial_{\tHat}{x})\, \partial_{x} +
(\partial_{\tHat} t)\, \partial_{t}
\label{intro_diffeo_partials_hat}
\end{align}
This gives the following
\begin{lemma}
$\morph^\pstar L_t\morph^\mstar$ and $\morph^\pstar L_t^2\morph^\mstar$ are given by
\begin{align}
\morph^\pstar L_t\morph^\mstar
&=
(\partial_t\tHat) L_\tHat + (\partial_t \xhat) L_\xhat
+
(\partial_\tHat\partial_t\tHat) d\tHat\wedge i_\tHat
+
(\partial_\tHat\partial_t\xhat) d\tHat\wedge i_\xhat
+
(\partial_\xhat\partial_t\tHat) d\xhat\wedge i_\tHat
+
(\partial_\xhat\partial_t\xhat) d\xhat\wedge i_\xhat
\label{intro_Lt}
\end{align}
and
\begin{equation}
\begin{aligned}
\morph^\pstar & L_t^2\morph^\mstar
=
\\&
  \left( \partial_t^2 \tHat \right) L_\tHat + \left( \partial_t^2 \xhat \right) L_\xhat
 \quad
 +
  \left( \partial_t \tHat \right)^2 L_\tHat^2
 \quad
 +
  2\left( \partial_t\tHat \right)\left( \partial_t \xhat \right) L_\tHat L_\xhat
 \quad
 +
  \left( \partial_t \xhat \right)^2 L_\xhat^2
\\&+
  2\left( \partial_t \tHat \right)\left( \partial_\tHat\partial_t\tHat \right)
  d\tHat\!\wedge\! i_\tHat L_\tHat
 \quad
 +
  2\left( \partial_t \tHat \right)\left( \partial_\tHat\partial_t\xhat \right)
  d\tHat\!\wedge\! i_\xhat L_\tHat
 \quad
 +
  2\left( \partial_t \tHat \right)\left( \partial_\xhat\partial_t\tHat \right)
  d\xhat\!\wedge\! i_\tHat L_\tHat
\\&+
  2\left( \partial_t \tHat \right)\left( \partial_\xhat\partial_t\xhat \right)
  d\xhat\!\wedge\! i_\xhat L_\tHat
 \quad
 +
  2\left( \partial_t \xhat \right)\left( \partial_\tHat\partial_t\tHat \right)
  d\tHat\!\wedge\! i_\tHat L_\xhat
 \quad
 +
  2\left( \partial_t \xhat \right)\left( \partial_\tHat\partial_t\xhat \right)
  d\tHat\!\wedge\! i_\xhat L_\xhat
\\&+
2\left( \partial_t \xhat \right)\left( \partial_\xhat\partial_t\tHat \right)
d\xhat\!\wedge\! i_\tHat L_\xhat
+
2\left( \partial_t \xhat \right)\left( \partial_\xhat\partial_t\xhat \right)
d\xhat\!\wedge\! i_\xhat L_\xhat
\\&+
  \left(
      \left( \partial_t\partial_\tHat\partial_t\tHat \right)
    + \left( \partial_\tHat\partial_t\tHat \right)^2
    + \left( \partial_\tHat\partial_t\xhat \right)
        \left( \partial_\xhat\partial_t\tHat \right)
  \right)
d\tHat\wedge i_\tHat
\\&+
  \left(
    \left( \partial_t\partial_\tHat\partial_t\xhat \right)
    + \left( \partial_\tHat\partial_t\tHat \right)
        \left( \partial_\tHat\partial_t\xhat \right)
    + \left( \partial_\tHat\partial_t\xhat \right)
        \left( \partial_\xhat\partial_t\xhat \right)
  \right)
  d\tHat\wedge i_\xhat
\\&+
  \left(
    \left( \partial_t\partial_\xhat\partial_t\tHat \right)
    + \left( \partial_\xhat\partial_t\tHat \right)
         \left( \partial_\tHat\partial_t\tHat \right)
    + \left( \partial_\xhat\partial_t\xhat \right)
         \left( \partial_\xhat\partial_t\tHat \right)
  \right)
  d\xhat\wedge i_\tHat
\\&+
  \left(
    \left( \partial_t \partial_\xhat\partial_t\xhat \right)
    + \left( \partial_\xhat\partial_t\tHat \right)
         \left( \partial_\tHat\partial_t\xhat \right)
    + \left( \partial_\xhat\partial_t\xhat \right)^2
  \right)
  d\xhat\wedge i_\xhat
\\&+
  2 
  \left( 
    \left( \partial_\tHat\partial_t\tHat \right)
    + \left( \partial_\xhat\partial_t\xhat \right)
    - \left( \partial_\tHat\partial_t\xhat \right)
    + \left( \partial_\xhat\partial_t\tHat \right)
  \right) 
  d\tHat\wedge d\xhat\wedge i_\xhat\,i_\tHat
\end{aligned}
\label{intro_Lt_sqr}
\end{equation}
\end{lemma}
\begin{proof} Dropping the $\morph^\pstar$ and $\morph^\mstar$ then
\eqref{intro_Lt} follows from
\begin{align*}
L_t
&=
L_{\partial_{t}}
=
L_{(\partial_t{\tHat})\, \partial_{\tHat}}
+
L_{(\partial_t{\xhat})\, \partial_{\xhat}}
=
(\partial_t{\tHat})\, L_{\partial_{\tHat}}
+
d (\partial_t{\tHat})\wedge i_{\partial_{\tHat}}
+
(\partial_t{\xhat})\, L_{\partial_{\xhat}}
+
d (\partial_t{\xhat})\wedge i_{\partial_{\xhat}}
\end{align*}
Taking $L_t$ of the first and second terms of \eqref{intro_Lt} we have
\begin{align*}
L_t \left( (\partial_t \tHat) L_\tHat \right) 
&=
(\partial_t^2 \tHat) L_\tHat
~~+
(\partial_t \tHat)L_t L_\tHat
\\&=
(\partial_t^2 \tHat) L_\tHat
\\
&\quad
~~+
  (\partial_t \tHat) 
  \Big[ 
    (\partial_t\tHat) L_\tHat
    ~~+
    (\partial_t \xhat) L_\xhat
    ~~+
    (\partial_\tHat\partial_t\tHat) d\tHat\wedge i_\tHat
\\
&\qquad
    ~~+
    (\partial_\tHat\partial_t\xhat) d\tHat\wedge i_\xhat
    ~~+
    (\partial_\xhat\partial_t\tHat) d\xhat\wedge i_\tHat
    ~~+
    (\partial_\xhat\partial_t\xhat) d\xhat\wedge i_\xhat
  \Big] 
  L_\tHat
\\&=
(\partial_t^2 \tHat) L_\tHat
~~+
(\partial_t \tHat)^2 L_\tHat^2
~~+
(\partial_t\tHat)(\partial_t \xhat) L_\tHat L_\xhat
~~+
(\partial_t \tHat)(\partial_\tHat\partial_t\tHat)
d\tHat\!\wedge\! i_\tHat L_\tHat
\\&\qquad
~~+
(\partial_t \tHat)(\partial_\tHat\partial_t\xhat)
d\tHat\!\wedge\! i_\xhat L_\tHat
~~+
(\partial_t \tHat)(\partial_\xhat\partial_t\tHat)
d\xhat\!\wedge\! i_\tHat L_\tHat
~~+
(\partial_t \tHat)(\partial_\xhat\partial_t\xhat)
d\xhat\!\wedge\! i_\xhat L_\tHat
\end{align*}
and
\begin{align*}
L_t \left( (\partial_t \xhat) L_\xhat \right) 
&=
(\partial_t^2 \xhat) L_\xhat
+
(\partial_t \xhat)L_t L_\xhat
\\&=
(\partial_t^2 \xhat) L_\xhat
+
(\partial_t \xhat) 
  \Big[ 
    (\partial_t\tHat) L_\tHat 
    ~~+ 
    (\partial_t \xhat) L_\xhat
    ~~+
    (\partial_\tHat\partial_t\tHat) d\tHat\wedge i_\tHat
\\
&\qquad\qquad\qquad\qquad
    +
    (\partial_\tHat\partial_t\xhat) d\tHat\wedge i_\xhat
    ~+
    (\partial_\xhat\partial_t\tHat) d\xhat\wedge i_\tHat
    ~+
    (\partial_\xhat\partial_t\xhat) d\xhat\wedge i_\xhat
  \Big] 
  L_\xhat
\\&=
(\partial_t^2 \xhat) L_\xhat
+
(\partial_t\tHat)(\partial_t \xhat) L_\tHat L_\xhat
+
(\partial_t \xhat)^2 L_\xhat^2
+
(\partial_t \xhat)(\partial_\tHat\partial_t\tHat)
d\tHat\!\wedge\! i_\tHat L_\xhat
\\&\qquad
+
(\partial_t \xhat)(\partial_\tHat\partial_t\xhat)
d\tHat\!\wedge\! i_\xhat L_\xhat
+
(\partial_t \xhat)(\partial_\xhat\partial_t\tHat)
d\xhat\!\wedge\! i_\tHat L_\xhat
+
(\partial_t \xhat)(\partial_\xhat\partial_t\xhat)
d\xhat\!\wedge\! i_\xhat L_\xhat
\end{align*}

Now
\begin{align*}
&L_t  \left( (\partial_\tHat\partial_t\tHat) d\tHat\wedge i_\tHat \right) 
=
(\partial_t\partial_\tHat\partial_t\tHat) d\tHat\wedge i_\tHat
~~+
(\partial_\tHat\partial_t\tHat) L_t \left( d\tHat\wedge i_\tHat \right) 
\\&=
(\partial_t\partial_\tHat\partial_t\tHat) d\tHat\wedge i_\tHat
\\&\quad+
(\partial_\tHat\partial_t\tHat)
  \Big[
      (\partial_t\tHat) L_\tHat
    + (\partial_t \xhat) L_\xhat
    + (\partial_\tHat\partial_t\tHat) d\tHat\wedge i_\tHat
\\
&\qquad\qquad\qquad
    + (\partial_\tHat\partial_t\xhat) d\tHat\wedge i_\xhat
    + (\partial_\xhat\partial_t\tHat) d\xhat\wedge i_\tHat
    + (\partial_\xhat\partial_t\xhat) d\xhat\wedge i_\xhat
  \Big] 
  \left( d\tHat\wedge i_\tHat \right) 
\end{align*}
The first term on the right hand side of the above is simply included
in \eqref{intro_Lt_sqr}. The next two terms double the appropriate
cross terms since
$L_\tHat d\xhat\wedge i_\xhat = d\xhat\wedge i_\xhat L_\tHat$.
The last for terms can be calculated according to the following table.
\begin{align*}
\begin{array}{|cc|c|c|c|c|}
\hline
&&\multicolumn{4}{c|}{\text{Second}}
\\
&& d\tHat\wedge i_\tHat & d\tHat\wedge i_\xhat & d\xhat\wedge i_\tHat & d\xhat\wedge i_\xhat
\\\hline
& d\tHat\wedge i_\tHat & d\tHat\wedge i_\tHat & d\tHat\wedge i_\xhat & 0 &
   d\tHat\wedge d\xhat\wedge i_\xhat\, i_\tHat
\\\cline{2-6}
& d\tHat\wedge i_\xhat & 0 & 0 & d\tHat\wedge i_\tHat -  d\tHat\wedge d\xhat\wedge i_\xhat\, i_\tHat
   & d\tHat \wedge i_\xhat
\\\cline{2-6}
\raisebox{0pt}[0pt]{\rotatebox{90}{\text{First}}}
& d\xhat\wedge i_\tHat & d\xhat\wedge i_\tHat & d\xhat\wedge i_\xhat - d\tHat\wedge d\xhat\wedge i_\xhat\,i_\tHat
  & 0 & 0
\\\cline{2-6}
& d\xhat\wedge i_\xhat & d\tHat\wedge d\xhat\wedge i_\xhat\,i_\tHat & 0 & d\xhat\wedge i_\tHat & d\xhat\wedge i_\xhat
\\\hline
\end{array}
\end{align*}
which gives
\begin{align*}
\Big(&
(\partial_\tHat\partial_t\tHat) d\tHat\wedge i_\tHat
+
(\partial_\tHat\partial_t\xhat) d\tHat\wedge i_\xhat
+
(\partial_\xhat\partial_t\tHat) d\xhat\wedge i_\tHat
+
(\partial_\xhat\partial_t\xhat) d\xhat\wedge i_\xhat
\Big)^2
\\&=
(\partial_\tHat\partial_t\tHat)^2 d\tHat\wedge i_\tHat
+
(\partial_\tHat\partial_t\tHat)(\partial_\tHat\partial_t\xhat)
d\tHat\wedge i_\xhat
+
(\partial_\tHat\partial_t\tHat)(\partial_\xhat\partial_t\xhat)
d\tHat\wedge d\xhat\wedge i_\xhat\, i_\tHat
\\&\quad +
(\partial_\tHat\partial_t\xhat)(\partial_\xhat\partial_t\tHat)
 \left( d\tHat\wedge i_\tHat -  d\tHat\wedge d\xhat\wedge i_\xhat\,i_\tHat \right) 
+
(\partial_\tHat\partial_t\xhat)(\partial_\xhat\partial_t\xhat)
d\tHat \wedge i_\xhat
\\&\quad +
(\partial_\xhat\partial_t\tHat) (\partial_\tHat\partial_t\tHat)
d\xhat\wedge i_\tHat
+
(\partial_\xhat\partial_t\tHat) (\partial_\tHat\partial_t\xhat)
 \left( d\xhat\wedge i_\xhat - d\tHat\wedge d\xhat\wedge i_\xhat\,i_\tHat \right) 
\\&\quad +
(\partial_\xhat\partial_t\xhat)(\partial_\tHat\partial_t\tHat)
d\tHat\wedge d\xhat\wedge i_\xhat\,i_\tHat
+
(\partial_\xhat\partial_t\xhat)(\partial_\xhat\partial_t\tHat)
d\xhat\wedge i_\tHat
+
(\partial_\xhat\partial_t\xhat)^2 d\xhat\wedge i_\xhat
\\&=
 \left( (\partial_\tHat\partial_t\tHat)^2
+ (\partial_\tHat\partial_t\xhat)(\partial_\xhat\partial_t\tHat) \right) 
d\tHat\wedge i_\tHat
+
 \left( 
(\partial_\tHat\partial_t\tHat)(\partial_\tHat\partial_t\xhat)
+(\partial_\tHat\partial_t\xhat) (\partial_\xhat\partial_t\xhat)
 \right) 
d\tHat\wedge i_\xhat
\\&\quad+
 \left( 
(\partial_\xhat\partial_t\tHat) (\partial_\tHat\partial_t\tHat)
+
(\partial_\xhat\partial_t\xhat)(\partial_\xhat\partial_t\tHat)
 \right) 
d\xhat\wedge i_\tHat
+
 \left( 
(\partial_\xhat\partial_t\tHat) (\partial_\tHat\partial_t\xhat)
+(\partial_\xhat\partial_t\xhat)^2
 \right) 
d\xhat\wedge i_\xhat
\\&\quad+
2 \left( 
(\partial_\tHat\partial_t\tHat)(\partial_\xhat\partial_t\xhat)
-(\partial_\tHat\partial_t\xhat)(\partial_\xhat\partial_t\tHat)
 \right) 
d\tHat\wedge d\xhat\wedge i_\xhat\,i_\tHat
\end{align*}
\end{proof}

Equations \eqref{tth_Lt}-\ref{tth_Lt_sqr} follow by setting
$(\partial_t\tHat)=1$ and $(\partial_x\tHat)=0$.

For the more general transformation \eqref{intro_diffeo_coords} then
$\Phi$ is transformed using \eqref{intro_Psihat_Phihat} with
\begin{align}
\morph^\pstar\,\Phi\,\morph^\mstar
=
\Plasf^2 \Big(
(\partial_\tHat t)(\partial_t\tHat) d\tHat\wedge i_\tHat
+
(\partial_\tHat t)(\partial_t\xhat) d\tHat\wedge i_\xhat
+
(\partial_\xhat t)(\partial_t\tHat) d\xhat\wedge i_\tHat
+
(\partial_\xhat t)(\partial_t\xhat) d\xhat\wedge i_\xhat\Big)
~.
\label{Appx_Phibar}
\end{align}

%
\end{document}